\documentclass[leqno,12pt]{article}

\usepackage{mathptmx} \usepackage[T1]{fontenc}
\usepackage[]{amsfonts} \usepackage[]{amsmath}
\usepackage[]{amssymb} \usepackage[]{latexsym}
\usepackage{graphicx,color} \usepackage{amsthm}[]
\usepackage{mathrsfs} \usepackage{url}
\usepackage[toc,page]{appendix}

\numberwithin{equation}{section}

\begin{document}

\newtheorem{theorem}{Theorem}[section]
\newtheorem{proposition}[theorem]{Proposition}
\newtheorem{lemma}[theorem]{Lemma}
\newtheorem{corollary}[theorem]{Corollary}

\theoremstyle{definition}
\newtheorem{definition}[theorem]{Definition}
\newtheorem{example}[theorem]{Example}
\newtheorem{remark}[theorem]{Remark}
\newtheorem{assumptions}[theorem]{Assumptions}

\newcommand{\ds}{\displaystyle} \newcommand{\nl}{\newline}
\newcommand{\eps}{\varepsilon}
\newcommand{\bE}{\mathbb{E}}
\newcommand{\cA}{\mathcal{A}}
\newcommand{\cB}{\mathcal{B}}
\newcommand{\cC}{\mathcal{C}}
\newcommand{\cL}{\mathcal{L}}
\newcommand{\fL}{\mathfrak{L}}
\newcommand{\cS}{\mathcal{S}}
\newcommand{\cO}{\mathcal{O}}
\newcommand{\cQ}{\mathcal{Q}}
\newcommand{\cH}{\mathcal{H}}
\newcommand{\cF}{\mathcal{F}}
\newcommand{\cM}{\mathcal{M}}
\newcommand{\cI}{\mathcal{I}}
\newcommand{\cD}{\mathcal{D}}
\newcommand{\cG}{\mathcal{G}}
\newcommand{\cT}{\mathcal{T}}
\newcommand{\cP}{\mathcal{P}}
\newcommand{\bP}{\mathbb{P}}
\newcommand{\bT}{\mathbb{T}}
\newcommand{\bD}{\mathbb{D}}
\newcommand{\bQ}{\mathbb{Q}}
\newcommand{\bC}{\mathbb{C}}
\newcommand{\bN}{\mathbb{N}}
\newcommand{\bR}{\mathbb{R}}
\newcommand{\II}{\mbox{II}}
\newcommand{\I}{\mbox{I}}
\newcommand{\QV}{\mbox{QV}}

\title{Functional It\^o Calculus, Path-dependence and the Computation of Greeks}

\author{Samy Jazaerli \and Yuri F. Saporito}

\author{Samy Jazaerli \thanks{Centre de Math\'ematiques Appliqu\'ees (CMAP), \'Ecole Polytechnique, Palaiseau, France, \texttt{samy.jazaerli@polytechnique.edu}} \and Yuri F. Saporito \thanks{Escola de Matem\'atica Aplicada (EMAp), Funda\c{c}\~ao Get\'ulio Vargas (FGV), Rio de Janeiro, Brazil, {\em yuri.saporito@fgv.br}.}}

\maketitle

\begin{abstract}

Dupire's functional It\^o calculus provides an alternative approach to the classical Malliavin calculus for the computation of sensitivities, also called Greeks, of path-dependent derivatives prices. In this paper, we introduce a measure of path-dependence of functionals within the functional It\^o calculus framework. Namely, we consider the Lie bracket of the space and time functional derivatives, which we use to classify functionals accordingly to their degree of path-dependence. We then revisit the problem of efficient numerical computation of Greeks for path-dependent derivatives using integration by parts techniques. Special attention is paid to path-dependent functionals with zero Lie bracket, called locally weakly path-dependent functionals in our classification. Hence, we derive the weighted-expectation formulas for their Greeks. In the more general case of fully path-dependent functionals, we show that, equipped with the functional It\^o calculus, we are able to analyze the effect of the Lie bracket on the computation of Greeks. Moreover, we are also able to consider the more general dynamics of path-dependent volatility. These were not achieved using Malliavin calculus.

\end{abstract}

\section{Introduction}

The theory of functional It\^o calculus, introduced in Dupire's seminal paper \cite{fito_dupire}, extends It\^o's stochastic calculus to functionals of the current history of a given process, and hence provides an excellent tool to study path-dependence. Further work extending this theory and its applications can be found in the partial list \cite{rama_cont_fito_mart, rama_cont_fito_formula, rama_cont_fito_change_variable, fito_zhang_1, fito_touzi_ppde1, fito_touzi_ppde2, fito_bsde_peng}.

We intuitively understand path-dependence of a functional as a measurement of its changes when the history of the underlying path varies. Here we propose a measure of path-dependence given by the Lie bracket of the space and time functional derivatives. Roughly, this is an instantaneous measure of path-dependence, since we consider only path perturbations at the current time. We then classify functionals based on this measure. Moreover, we analyze the relation of what we called \textit{locally weakly path-dependent} functionals and the Monte Carlo computation of Greeks in path-dependent volatility models, cf. \cite{malliavin_greeks1}.

Malliavin calculus was successfully applied to derive these Monte Carlo procedures to compute Greeks of path-dependent derivatives in local volatility models, see for example \cite{malliavin_greeks1, malliavin_greeks2, malltelm05, gobet_malliavin_barrier_lookback, gobet_revisiting_greeks, nualart_malliavin_book}. However, the theory presented here allows us to extend these Monte Carlo procedures to a wider class of path-dependent derivatives provided that the path-dependence is not too severe. This will be made precise in Section \ref{sec:weakly_path_depend_section}. We will also see that the functional It\^o calculus can be used to derive the same weighted-expectation formulas shown in \cite{malliavin_greeks1}.

Furthermore, unlike the Malliavin calculus approach, we are also able to provide a formula for the Delta of functionals with more severe path-dependence, here called \textit{strongly path-dependent}. In its current form, this formula enhances the understanding of the weights for different cases of path-dependence, although it is not as computationally appealing as the ones derived for locally weakly path-dependent functionals. It shows however the impact that the Lie bracket has on the Delta of a derivative contract. Additionally, the functional It\^o calculus allows us to consider the more general path-dependent volatility models, see \cite{foschi2008path}, \cite{guyon_path_vol} and \cite{complete_sv_rogers}, for example.

Our main contribution is the introduction of a measure of path-dependence and the application of such measure to the computation of Greeks for path-dependent derivatives.

The paper is organized as follows. In Section \ref{sec:fito_primer}, we provide some background on functional It\^o calculus. Section \ref{sec:weakly_path_depend_section} introduces the measure of path-dependence and a classification of functionals accordingly to this measure. Finally, we present applications of this measure of path-dependence to the computation of Greeks in Section \ref{sec:greeks_sec}. Two numerical examples, related to Asian options and quadratic variation contracts, are discussed.

\section{A Primer on Functional It\^o Calculus}\label{sec:fito_primer}

In this section we will present some definitions and results of the functional It\^o calculus that will be necessary in Sections \ref{sec:weakly_path_depend_section}  and \ref{sec:greeks_sec}.

The space of $\bR$-valued c\`adl\`ag paths in $[0,t]$ will be denoted by $\Lambda_t$. We also fix a time horizon $T > 0$. The \textit{space of paths} is then defined as
$$\Lambda = \bigcup_{t \in [0,T]} \Lambda_t.$$

\textit{A very important remark on the notation:} as in \cite{fito_dupire}, we will denote elements of $\Lambda$ by upper case letters and often the final time of its domain will be subscripted, e.g. $Y \in \Lambda_t \subset \Lambda$ will be denoted by $Y_t$. Note that, for any $Y \in \Lambda$, there exists only one $t$ such that $Y \in \Lambda_t$. The value of $Y_t$ at a specific time will be denoted by lower case letter: $y_s = Y_t(s)$, for any $s \leq t$. Moreover, if a path $Y_t$ is fixed, the path $Y_s$, for $s \leq t$, will denote the restriction of the path $Y_t$ to the interval $[0,s]$.

The following important path operations are always defined in $\Lambda$. For $Y_t \in \Lambda$ and $t \leq s \leq T$, the \textit{flat extension} of $Y_t$ up to time $s \geq t$ is defined as
$$Y_{t,s-t}(u) = \left\{
\begin{array}{ll}
  y_u,  &\mbox{ if } \quad 0 \leq u \leq t, \\
  y_t,  &\mbox{ if } \quad t \leq u \leq s,
\end{array}
\right.$$
see Figure \ref{fig:flat}. For $h \in \bR$, the \textit{bumped path} $Y_t^h$, shown in Figure \ref{fig:bump}, is defined by
$$Y_t^h(u) = \left\{
\begin{array}{ll}
  y_u,      &\mbox{ if } \quad 0 \leq u < t, \\
  y_t + h,  &\mbox{ if } \quad u = t.
\end{array}
\right.$$
\begin{figure}[h!]
  \begin{minipage}[b]{0.5\linewidth}
    \centering
    \includegraphics[width=0.8\linewidth]{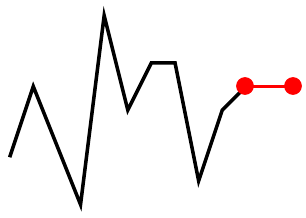}
    \caption{Flat extension of a path.}
    \label{fig:flat}
  \end{minipage}
  \hspace{0.5cm}
  \begin{minipage}[b]{0.5\linewidth}
    \centering
    \includegraphics[width=0.8\linewidth]{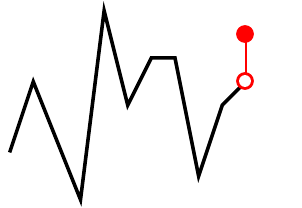}
    \caption{Bumped path.}
    \label{fig:bump}
  \end{minipage}
\end{figure}

For any $Y_t, Z_s \in \Lambda$, where it is assumed without loss of generality that $s \geq t$, we define the following metric in $\Lambda$,
$$d_{\Lambda}(Y_t,Z_s) = \| Y_{t,s-t} - Z_s\|_{\infty} + |s -t|,$$
where
$$\|Y_t\|_{\infty} = \sup_{u \in [0,t]} |y_u|.$$

A \textit{functional} is any function $f:\Lambda \longrightarrow \bR$ and it is said to be $\Lambda$-continuous if it is continuous with respect to the metric $d_{\Lambda}$.

Moreover, for a functional $f$ and a path $Y_t$ with $t < T$, if the following limit exists, the \textit{time functional derivative} of $f$ at $Y_t$ is defined as
\begin{align}
\Delta_t f(Y_t) = \lim_{\delta t \to 0^+} \frac{f(Y_{t,\delta t}) - f(Y_t)}{\delta t}. \label{eq:time_deriv}
\end{align}
The \textit{space functional derivative} of $f$ at $Y_t$ is defined as
\begin{align}
\Delta_x f(Y_t) = \lim_{h \to 0} \frac{f(Y_t^h) - f(Y_t)}{h}, \label{eq:space_deriv}
\end{align}
when this limit exists, and for this derivative it is allowed $t = T$.

Additionally, we say a functional $f$ is \textit{boundedness-preserving} if, for every compact set $K \subset \bR$, there exists a constant $C$ such that $|f(Y_t)| \leq C$, for every path $Y_t$ satisfying $Y_t([0,t]) = \{y \in \bR \ ; \ Y_t(s) = y \mbox{ for some } s \in [0,t]\} \subset K$.

Finally, a functional $f: \Lambda \longrightarrow \bR$ is said to be in $\bC^{1,2}$ if it is $\Lambda$-continuous, boundedness-preserving and it has $\Lambda$-continuous, boundedness-preserving derivatives $\Delta_t f$, $\Delta_x f$ and $\Delta_{xx} f$. With obvious definition, we also use the notation $\bC^{i,j}$, where $\bC = \bC^{0,0}$ is the space of $\Lambda$-continuous functions.

Before continuing, some comments about conditional expectation in the context of paths and functionals. Until now, we have not considered any probability framework. We then fix throughout the paper a probability space $(\Omega, \cF, \bP)$. For any $s \leq t$ in $[0,T]$, denote by $\Lambda_{s,t}$ the space of $\bR$-valued c\`adl\`ag paths on $[s,t]$. Now define the operator $(\cdot \ \otimes \ \cdot) : \Lambda_{s,t} \times \Lambda_{t,T} \longrightarrow \Lambda_{s,T}$, the \textit{concatenation} of paths, by
$$(Y \otimes Z)(u) = \left\{
\begin{array}{ll}
  y_u,  &\mbox{ if } s \leq u < t, \\
  z_u - z_t + y_t, &\mbox{ if }  t \leq u \leq T,
\end{array}
\right.$$
which is a paste of $Y$ and $Z$.

Given functionals $\mu$ and $\sigma$, we consider a process $x$ given by the Stochastic Differential Equation (SDE)
\begin{align}
dx_s = \mu(X_s)ds + \sigma(X_s)dw_s, \label{eq:sde}
\end{align}
with $s \geq t$ and $X_t = Y_t$. The process $(w_s)_{s \in [0,T]}$ denotes a standard Brownian motion in $(\Omega, \cF, \bP)$ and we assume there exists a unique strong solution for the SDE (\ref{eq:sde}). This unique solution will be denoted by $x_s^{Y_t}$ and the path solution from $t$ to $T$ by $X_{t,T}^{Y_t}$. We forward the reader, for instance, to \cite{rogerswilliams} for results on SDEs with functional coefficients.

\begin{remark}[Strong Solutions]\label{rmk:strong_solution}
Unique strong solution of (\ref{eq:sde}) might be achieved by requiring that $\mu$ and $\sigma$ are in $\bC$ and satisfy the usual (fixed-time) Lipschitz and linear growth conditions:
\begin{align}
&|\mu(Y_t) - \mu(Z_t)| + |\sigma(Y_t) - \sigma(Z_t)| \leq K \|Y_t - Z_t\|_{\infty},\\
&|\mu(Y_t)| + |\sigma(Y_t)| \leq K(1 + \|Y_t\|_{\infty}),
\end{align}
for all $Y_t, Z_t \in \Lambda$, where $K > 0$ is constant. The continuity of $\mu$ and $\sigma$ guarantee the proper measurability conditions, see Section \ref{sec:int_by_parts}
\end{remark}

Finally, we define the \textit{conditioned expectation} as
\begin{align}
\bE[g(X_T) \ | \ Y_t] = \bE[g(Y_t \otimes X_{t,T}^{Y_t})], \label{eq:conditioned_expec}
\end{align}
for any $Y_t \in \Lambda$. The path $Y_t \otimes X_{t,T}^{Y_t} \in \Lambda_T$ is equal to the path $Y_t$ up to $t$ and follows the dynamics of the SDE (\ref{eq:sde}) from $t$ to $T$ with initial path $Y_t$. Moreover, if we define the filtration $\cF_t^x$ generated by $\{x_s \ ; \ s \leq t\}$, one may prove
$$\bE[g(X_T) \ | \ X_t(\omega)] = \bE[g(X_T) \ | \ \cF_t^x](\omega) \quad \bP\mbox{-a.s.}$$
where the expectation on the left-hand side is the one discussed above and the one on the right-hand side is the usual \textit{conditional expectation}.

An interesting issue regarding conditioned expectation is to study its smoothness within the functional It\^o calculus framework. It would clearly depend on the smoothness of the functional $g$. A more intricate dependence would be with respect to the process $x$ and its coefficients. A partial answer is given in \cite{fito_bsde_peng}, where the authors derived conditions on $g$ so that the conditioned expectation operator belongs to $\bC^{1,2}$ in the Brownian motion case.

For the sake of completeness the functional It\^o formula is stated here. The proof can be found in \cite{fito_dupire}.
\begin{theorem}[Functional It\^o Formula; \cite{fito_dupire}]\label{thm:fif}
Let $x$ be a continuous semimartingale and $f \in \bC^{1,2}$. Then, for any $t \in [0,T]$,
$$f(X_t) = f(X_0) + \int_0^t \Delta_t f(X_s) ds + \int_0^t \Delta_x f(X_s) dx_s + \frac{1}{2} \int_0^t \Delta_{xx} f(X_s) d\langle x\rangle_s \quad \bP\mbox{-a.s.}$$
\end{theorem}

\subsection{An Integration by Parts Formula for $\Delta_x$}\label{sec:int_by_parts}

In this section, we present some results from \cite{rama_cont_fito_mart} regarding the adjoint of $\Delta_x$. Fix a continuous square-integrable martingale $(x_t)_{t \in [0,T]}$ and the filtration generated by it, $(\cF^x_t)_{t \in [0,T]}$.

We denote the space of continuous square-integrable martingales in $[0,T]$ with respect to the filtration $(\cF^x_t)_{t \in [0,T]}$ by $\cM^2_c$ and we define
\begin{align}
H^2_x & = \left\{ f \in \bC \ ; \  \bE\left[\int_0^T f^2(X_t) d\langle x \rangle_t \right] < +\infty \right\}, \label{eq:H2}\\
L^2_{loc,x} & = \left\{ f \in \bC \ ; \  \int_0^T f^2(X_t) d\langle x \rangle_t  < +\infty \ \ \bP-\mbox{a.s. } \right\}, \label{eq:L2_loc}\\
M^2_x & = \Big\{ f \in \bC \ ; \ (f(X_t))_{t \in [0,T]} \in \cM^2_c \Big\}. \label{eq:M2}
\end{align}
We could consider more general measurability conditions on $f$ to define the spaces above. However, $\Lambda$-continuity of the functional and continuity of the process $x$ guarantee the required measurability in order to consider stochastic integrals of $f(X)$ with respect to $x$, namely $(f(X_t))_{t \in [0,T]}$ will be progressively measurable with respect to $(\cF^x_t)_{t \in [0,T]}$.

Consider now the inner products
\begin{align}
\langle f , g \rangle_{\cH^2_x} & = \bE \left[ \int_0^T f(X_t) g(X_t) d\langle x \rangle_t \right], \label{eq:ip_H2} \\
\langle f , g \rangle_{\cM^2_x} & = \bE \left[ f(X_T) g(X_T) \right], \label{eq:ip_M2}
\end{align}
in $H^2_x$ and $M_x^2$, respectively. So that (\ref{eq:ip_H2}) and (\ref{eq:ip_M2}) are proper inner products, it is necessary to suitably identify elements of these spaces as follows:
$$f \sim g \Leftrightarrow f(X_t) = g(X_t), \quad d\langle x\rangle_t \times d\bP.$$
Thus the quotient spaces $\cH^2_x = H^2_x/\sim$ and $\cM^2_x = M^2_x/\sim$ are both Hilbert spaces. 

\begin{remark}\label{rmk:L2_loc}
Notice that since we are considering $f$ $\Lambda$-continuous, then we clearly have $f \in L^2_{loc,x}$.
\end{remark}

Define now the It\^o integral operator $\cI_x: \cH^2_x \rightarrow \cM^2_x$ as
$$\cI_x(f)(t) = \int_0^t f(X_s) dx_s,$$
which is an isometry. Indeed,
\begin{align}
\langle f, g \rangle_{\cH^2_x} = \langle \cI_x(f), \cI_x(g) \rangle_{\cM^2_x}. \label{eq:isometry}
\end{align}

A \textit{test functional} is an element of
\begin{align}
\cD_x = \{ f \in \bC^{1,2} \cap \cM^2_x \ ; \ \Delta_x f \in \cH^2_x\}. \label{eq:test_functional}
\end{align}
The next proposition describes the integration by parts formula of the operator $\Delta_x$ in the space $\cD_x$.

\begin{proposition}[\cite{rama_cont_fito_mart}]
For any $f \in \cD_x$ and $g \in \cH^2_x$,
\begin{align}
\langle \Delta_xf , g \rangle_{\cH^2_x} = \langle f, \cI_x(g) \rangle_{\cM^2_x}. \label{eq:ibp}
\end{align}
\end{proposition}

\begin{proof}
Since $\bE[\cI_x(g)] = 0$ and the goal is to compute $\Delta_x f$, it can be assumed without loss of generality that $f(X_0) = 0$. Then, by the Functional It\^o Formula, Theorem \ref{thm:fif},
\begin{align*}
\cI_x(\Delta_x f)(t) &= \int_0^t \Delta_x f(X_s) dx_s = f(X_t) - \int_0^t \Delta_t f(X_s) ds - \frac{1}{2} \int_0^t \Delta_{xx} f(X_s) d\langle x \rangle_s,
\end{align*}
and thus, since $f \in \cM^2_x$, by the uniqueness of the semimartingale decomposition,
$$\int_0^t \Delta_t f(X_s) ds + \frac{1}{2} \int_0^t \Delta_{xx} f(X_s) d\langle x \rangle_s = 0.$$
Therefore
$$\cI_x(\Delta_x f)(t) = f(X_t),$$
which implies the integration by parts formula:
$$\langle \Delta_xf , g \rangle_{\cH^2_x} = \langle \cI_x(\Delta_x f), \cI_x(g) \rangle_{\cM^2_x} = \langle f, \cI_x(g) \rangle_{\cM^2_x},$$
for all $f \in \cD_x$ and $g \in \cH^2_x$, where we have used It\^o Isometry (\ref{eq:isometry}).
\end{proof}

More generally, as it was shown in \cite{rama_cont_fito_mart}, the operator $\Delta_x$ is closable in $\cM^2_x$ and its adjoint is the It\^o integral.

\subsection{Path-Dependent PDE}\label{sec:pdv}

Suppose that the dynamics of a stock price $x$, under a risk-neutral measure, is given by the path-dependent volatility model (\cite{foschi2008path}, \cite{guyon_path_vol} and \cite{complete_sv_rogers}, for instance),
\begin{align}
dx_t = r x_t dt + \sigma(X_t)dw_t. \label{eq:PD_vol}
\end{align}
So, a no-arbitrage price of a path-dependent derivative with maturity $T$ and payoff given by the functional $g: \Lambda_T \longrightarrow \bR$, which will be called \textit{contract}, is given by
$$f(Y_t) = e^{-r(T-t)}\bE\left[\left. g(X_T) \ \right| \ Y_t \right],$$
see Equation (\ref{eq:conditioned_expec}) for the exact definition of this quantity. This expectation is taken under the chosen risk-neutral measure. Finally, we state the path-dependent extension of the pricing Partial Differential Equation (PDE), which is acronymed PPDE; see for instance \cite{fito_dupire, fito_zhang_1, fito_touzi_ppde1, fito_touzi_ppde2, fito_bsde_peng}.

\begin{theorem}[Pricing PPDE; \cite{fito_dupire}]\label{thm:functional_pde_theorem}
If the price of a path-dependent derivative with contract $g$, denoted by the functional $f$, belongs to $\bC^{1,2}$, then, for any $Y_t$ in the topological support of the process $x$,
\begin{align}\label{eq:functional_pde}
\Delta_t f(Y_t) + \frac{1}{2} \sigma^2(Y_t) \Delta_{xx} f(Y_t) + r y_t \Delta_x f(Y_t) - r f(Y_t) = 0,
\end{align}
with final condition $f(Y_T) = g(Y_T)$.
\end{theorem}

\begin{remark}\label{rmk:stroock_varadhan}
In local volatility models of \cite{dupire94} ($\sigma(Y_t) = \sigma(t,y_t)$), under mild assumptions on $\sigma$, the Stroock-Varadhan Support Theorem states that the topological support of $x$ is the space of continuous paths starting at $x_0$, see for instance \cite[Chapter 2]{pinsky95}. So, under these assumptions, the PPDE (\ref{eq:functional_pde}) will hold for any continuous path. See Appendix \ref{sec:topological_support} for a discussion on this type of result in the case of SDEs of the form (\ref{eq:PD_vol}).
\end{remark}

\section{Path-Dependence}\label{sec:weakly_path_depend_section}

The goal of this section is to analyze the commutation issue of the operators $\Delta_x$ and $\Delta_t$. To start, consider the following example
$$\I(Y_t) = \int_0^t y_u du.$$
A simple computation shows
$$\Delta_t \I(Y_t) = y_t \mbox{ and } \Delta_x \I(Y_t) = 0,$$
and hence
$$\Delta_x(\Delta_t \I)(Y_t) = 1 \neq 0 = \Delta_t(\Delta_x \I)(Y_t).$$
On the other hand, it is clear that the operators commute when applied to functionals of the form $f(Y_t) = \phi(t,y_t)$, where $\phi$ is smooth. Therefore, one could ask if the operators commute for a functional $f$ if and only if $f$ is of the form $\phi(t,y_t)$. The following counter-example shows that this is not true. Consider
\begin{align}
\II(Y_t) = \int_0^t \int_0^s y_u du ds, \label{eq:double_integral}
\end{align}
and then notice
$$\Delta_t \II(Y_t) = \int_0^t y_s ds \mbox{ and } \Delta_x \II(Y_t) = 0,$$
which clearly implies that
$$\Delta_x(\Delta_t \II)(Y_t) = 0 = \Delta_t(\Delta_x \II)(Y_t).$$

\begin{definition}[Lie Bracket]\label{def:lie_bracket}
The \textit{Lie bracket (or commutator)} of the operators $\Delta_t$ and $\Delta_x$ will play a fundamental role in what follows and it is defined as
$$\fL f(Y_t) = [\Delta_x, \Delta_t] f(Y_t) = \Delta_{xt} f(Y_t) - \Delta_{tx} f(Y_t),$$
where $\Delta_{xt} = \Delta_x \Delta_t $ and $f$ is such that all the derivatives above exist. Abusing the nomenclature, we will call the operator $\fL$ by simply Lie bracket.
\end{definition}

The following lemma gives an alternative definition for the Lie bracket. For its proof, we will assume the technical assumption on $f$:
{\fontsize{10}{10}\selectfont
\begin{align}\label{eq:tech_assumption}
\lim_{h \to 0} \frac{f\left( (Y_{t,\delta t})^h\right) - f(Y_{t,\delta t}) - f(Y_t^h) + f(Y_t)}{h \delta t} = \frac{\Delta_x f(Y_{t,\delta t}) - \Delta_x f(Y_t)}{\delta t} \mbox{ uniformly in } \delta t.
\end{align}}
\begin{lemma}\label{lem:lie_bracket_limit}
Consider a functional $f$ such that $\fL f$ exists as in Definition \ref{def:lie_bracket} and that Condition (\ref{eq:tech_assumption}) is satisfied. Then, the Lie bracket of a functional $f$ is given by the following limit,
$$\fL f(Y_t) = \lim_{\delta t \to 0^+ \atop h \to 0} \frac{ f((Y_{t,\delta t})^h) - f((Y_t^h)_{t,\delta t})}{h \delta t}.$$
\end{lemma}

\begin{proof}
Firstly, notice that, since $\fL f$ exists,
\begin{align*}
\Delta_t \Delta_x f(Y_t) &= \lim_{\delta t \to 0^+} \lim_{h \to 0} \frac{f\left( (Y_{t,\delta t})^h\right) - f(Y_{t,\delta t}) - f(Y_t^h) + f(Y_t)}{h \delta t},\\
\Delta_x \Delta_tf(Y_t) &= \lim_{h \to 0} \lim_{\delta t \to 0^+}  \frac{f\left( (Y_t^h)_{t,\delta t}\right) - f(Y_{t,\delta t}) - f(Y_t^h) + f(Y_t)}{h \delta t}
\end{align*}
Now, by Condition (\ref{eq:tech_assumption}), the famous result by Moore about interchanging limit of functions (see \cite{moore_interchanging_limit}) can be employed and the result follows.
\end{proof}

This lemma gives a very interesting interpretation of the Lie bracket: it is a measure of instantaneous path-dependence of the functional $f$, i.e. it will be zero if, in the limit, the order of bump and flat extension of the path at the current time makes no difference. In Figure \ref{fig:bracket}, the term $(Y_{t,\delta t})^h$ is indicated in blue and the term $(Y_t^h)_{t,\delta t}$, in red. Lemma \ref{lem:lie_bracket_limit} also shows that the commutation issue for functionals is not just lack of smoothness as in the finite-dimensional case.
\begin{figure}[h!]
\begin{center}
\includegraphics{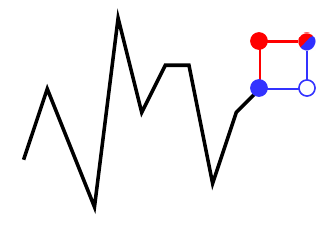} \caption{Geometric Interpretation of the $\fL.$} \label{fig:bracket}
\end{center}
\end{figure}

\begin{proposition}\label{prop:lie_bracket_chain_rule}
Suppose the functional $f:\Lambda \longrightarrow \bR$ is given by $f(X_t) = \phi(t, f_1(X_t),$ $\ldots, f_k(X_t))$, where $\phi: \bR_+ \times \bR^k \longrightarrow \bR$ has all the first and second order partial derivatives and the Lie bracket of $f_i$ exists for any $i=1,\ldots,k$. Then
\begin{align*}
\fL f(X_t) = \sum_{i=1}^k \frac{\partial \phi}{\partial x_i}(t, f_1(X_t), \ldots, f_k(X_t))\fL f_i(X_t)
\end{align*}
\end{proposition}

\begin{proof}
This follows easily by direct computation. Notice
\begin{align*}
\Delta_x f(X_t) &= \sum_{i=1}^k \frac{\partial \phi}{\partial x_i} \Delta_x f_i(X_t), \\
\Delta_t f(X_t) &= \frac{\partial \phi}{\partial t} + \sum_{i=1}^k \frac{\partial \phi}{\partial x_i} \Delta_t f_i(X_t).
\end{align*}
Hence, one concludes
\begin{align*}
\Delta_t\Delta_x f(X_t) &= \sum_{i=1}^k  \frac{\partial^2 \phi}{\partial x_i \partial t} \Delta_x f_i(X_t) + \frac{\partial \phi}{\partial x_i} \Delta_t \Delta_x f_i(X_t)\\
&+ \sum_{i=1}^k\sum_{j=1}^k \frac{\partial^2 \phi}{\partial x_i \partial x_j} \Delta_x f_i(X_t) \Delta_t f_j(X_t), \\
\Delta_x\Delta_t f(X_t) &= \sum_{i=1}^k \frac{\partial^2 \phi}{\partial x_i \partial t} \Delta_x f_i(X_t) + \frac{\partial \phi}{\partial x_i} \Delta_x \Delta_t f_i(X_t)\\
&+ \sum_{i=1}^k \sum_{j=1}^k \frac{\partial \phi}{\partial x_i \partial x_j} \Delta_x f_j(X_t) \Delta_t f_i(X_t).
\end{align*}
\end{proof}

\subsection{Classification of Path-Dependence of Functionals}

Based on the Lie bracket of $\Delta_t$ and $\Delta_x$, we define several different categories of path-dependence for functionals.

\begin{definition}\label{def:classif}
A functional $f:\Lambda \longrightarrow \bR$ is called
\begin{itemize}
\item \textit{locally weakly path-dependent} if $\fL f = 0$; \\
\item \textit{path-independent} if there exists $\phi:\bR_+ \times \bR \longrightarrow \bR$ such $f(Y_t) = \phi(t,y_t)$; \\
\item\textit{discretely monitored} if there exist $0 < t_1 < \cdots < t_n \leq T$ and, for each $t \in [0,T]$, $\phi(t): \bR^{i(t)} \longrightarrow \bR$ such that
\begin{align}
f(Y_t) = \phi(t, y_{t_1}, \ldots, y_{t_{i(t)}},y_t),
\end{align}
where $i(t)$ is the maximum $i \in \{1,\ldots,n\}$ such that $t_i \leq t$; \\
\item $t_1$\textit{-delayed path-dependent} if $\fL f(Y_t) = 0, \ \forall \ t < t_1$. Moreover, a functional $f$ is said to be \textit{delayed path-dependent} if there exists $t_1 > 0$ such that $f$ is $t_1$-delayed path-dependent; \\
\item \textit{strongly path-dependent} if $\forall \ [s,t] \subset [0,T],  \ \exists \ u \in [s,t], \ \fL f(Y_u) \neq 0$.
\end{itemize}
\end{definition}

\begin{remark}

In Mathematical Finance, the terminology \textit{weakly path-dependent} was used to denominate derivative prices that are solution of the classical Black--Scholes PDE with some additional boundary conditions, like, for example, American Vanilla options and barrier options. Assuming that the events of interest of these contracts have not happened, their prices are still functions of just time and the current value of the stock; see, for instance, \cite{wilmott_quant_fin}. We would like to advert the reader that this meaning of the terminology \textit{weakly path-dependent} has no relation with our definition. Here we would like to emphasise the term \textit{locally} in our terminology, which stress the instantaneous aspect of the Lie bracket $\fL$.

\end{remark}

The next proposition analyzes the Lie-bracket of discretely monitored functionals.

\begin{proposition}\label{prop:lie_bracket_discret}
If $f$ is a discretely monitored functional such that its Lie bracket exists, then $\fL f(Y_t) = 0$ but for $t_1, \ldots, t_n$.
\end{proposition}

\begin{proof}
Take $t \in (t_i, t_{i+1})$. So, for sufficiently small $\delta t > 0$ such that $t + \delta t \in (t_i, t_{i+1})$, we must have $f((Y_{t,\delta t})^h) = \phi(t+\delta t, y_{t_1}, \ldots, y_{t_{i(t)}},y_t+h) = f((Y_t^h)_{t,\delta t})$. Hence, $\fL f(Y_t) = 0$.
\end{proof}

\subsection{The Impact of the Contract Functional on the \\
Path-Dependence of the Price Functional}

In this section, we will present some investigatory discussion on how the contract functional $g: \Lambda_T \longrightarrow \bR$ influences the path-dependence, as measured by the Lie bracket, of the price $f$. We do not aim to have the most general assumptions on $g$.

The goal here is to connect a measure of path-dependence of the contract functional $g$ to the path-dependence, as measure by the Lie bracket, of the price functional $f$. In particular, we will derive a condition on $g$ alone such that the price $f$ is locally weakly path-dependent. From this, it is clear that we should restrict ourselves to Markovian dynamics for $x$. We then may write
$$f(Y_t) = \bE[g(X_T) \ | \ Y_t] = \bE[g(Y_t \otimes X_{t,T}^{y_t})],$$
where $X_{t,T}^{y_t}$ is the path from $t$ to $T$ followed by $x$ starting at $(t,y_t)$. Therefore,
\begin{align*}
f((Y_t^h)_{t,\delta t}) &= \bE[g((Y_t^h)_{t,\delta t} \otimes X_{t+\delta t,T}^{y_t+h})],\\
f((Y_{t,\delta t})^h) &= \bE[g((Y_{t,\delta t})^h \otimes X_{t+\delta t,T}^{y_t+h})].
\end{align*}
If $f$ satisfies Lemma \ref{lem:lie_bracket_limit}, we find
$$\fL f(Y_t) = \lim_{\delta t \to 0^+ \atop h \to 0} \bE\left[\frac{g((Y_t^h)_{t,\delta t} \otimes X_{t+\delta t,T}^{y_t+h}) - g((Y_{t,\delta t})^h \otimes X_{t+\delta t,T}^{y_t+h})}{h \delta t} \right].$$
We have then the following result, with straightforward proof.
\begin{proposition}
Let $g: \Lambda_T \longrightarrow \bR$ be a contract functional such that, for any $Y_t \in \Lambda$ and $Z_T \in \Lambda_T$, the following limit exists:
\begin{align}\label{eq:g_condition}
\phi(Y_t, Z_T) = \lim_{\delta t \to 0^+ \atop h \to 0} \frac{g((Y_t^h)_{t,\delta t} \otimes Z_{t+\delta t,T}^{y_t+h}) - g((Y_{t,\delta t})^h \otimes Z_{t+\delta t,T}^{y_t+h})}{h \delta t},
\end{align}
and the following boundedness assumption is satisfied:
\begin{align}\label{eq:g_condition_bound}
\left|g((Y_t^h)_{t,\delta t} \otimes Z_{t+\delta t,T}^{y_t+h}) - g((Y_{t,\delta t})^h \otimes Z_{t+\delta t,T}^{y_t+h})\right| \leq c(Y_t, Z_T) \psi(h, \delta t),
\end{align}
where $c: \Lambda \times \Lambda_T \longrightarrow [0,+\infty)$ with $\bE[c(Y_t, X_T)] < +\infty$ and the limit
\begin{align}\label{eq:g_condition_phi}
\lim_{\delta t \to 0^+ \atop h \to 0} \frac{\psi(h, \delta t)}{h \delta t}
\end{align}
exists. Then, if $f$ satisfies Lemma \ref{lem:lie_bracket_limit}, we find $\fL f(Y_t) = \bE[\phi(Y_t, X_T)]$.
\end{proposition}

In our case, since $x$ is a continuous diffusion, we may restrict the limit (\ref{eq:g_condition}) to continuous $Z_T$. Readily, if $\phi \equiv 0$, $f$ will be locally weakly path-dependent. Additionally, if the limit in Equation (\ref{eq:g_condition}) is a constant (with respect to $Z$) different than zero, then $f$ will not be locally weakly path-dependent.

We would like to point out that if $\phi \equiv 0$ in Equation (\ref{eq:g_condition}), then $f$ will be locally weakly path-dependent for any Markovian model for $x$. It should be clear that, in the case of path-dependent dynamics for $x$, the argument above does not work.

%

Below we analyze two interesting examples of contract functionals $g$ and the path-dependence of its price functional $f$ under Markovian models.

\begin{example}

Let us consider the example of the double integral, see Equation (\ref{eq:double_integral}):
\begin{align*}
\I(Y_t) = \int_0^t y_s ds \mbox{ and } \II(Y_t) = \int_0^t \int_0^s y_u du ds.
\end{align*}
We have seen that $\fL \II = 0$. However, as we will verify, $f(Y_t) = \bE[\varphi(\II(X_T)) \ | \ Y_t]$ might not be locally weakly path-dependent, depending on $\varphi$, that we assume to be in $C^1(\bR)$ with bounded derivative.
By direct computation, we have
\begin{align*}
&\II((Y_t^h)_{t,\delta t} \otimes Z_{t+\delta t,T}^{y_t+h}) - \II((Y_{t,\delta t})^h \otimes Z_{t+\delta t,T}^{y_t+h}) \\
&= h \int_t^{t+\delta t} (s-t) ds + h (T - t - \delta t) \delta t= h \delta t (T-t)-\frac{h \delta t^2}{2}
\end{align*}
which implies
{\fontsize{11}{11}\selectfont\begin{align*}
\varphi(\II((Y_t^h)_{t,\delta t} \otimes Z_{t+\delta t,T}^{y_t+h})) &- \varphi(\II((Y_{t,\delta t})^h \otimes Z_{t+\delta t,T}^{y_t+h})) = \varphi'(c) \left(h \delta t (T-t)-\frac{h \delta t^2}{2}\right),
\end{align*}}
for some $c$ between $\II((Y_t^h)_{t,\delta t} \otimes Z_{t+\delta t,T}^{y_t+h})$ and $\II((Y_{t,\delta t})^h \otimes Z_{t+\delta t,T}^{y_t+h})$. Therefore,
{\fontsize{11}{11}\selectfont\begin{align*}
\lim_{\delta t \to 0^+ \atop h \to 0} &\frac{\varphi(\II((Y_t^h)_{t,\delta t} \otimes Z_{t+\delta t,T}^{y_t+h})) - \varphi(\II((Y_{t,\delta t})^h \otimes Z_{t+\delta t,T}^{y_t+h}))}{h \delta t} = (T-t)\varphi'(\II(Y_t \otimes Z_{t,T})),
\end{align*}}
which means that $f$ is not locally weakly path-dependent, if $\varphi$ satisfies $\bE[\varphi'(\II(Y_t \otimes X_{t,T}))] \neq 0$. In this case, the locally weakly path-dependence property depends on the dynamics $x$.

\end{example}

\begin{example}

Let us consider now another example: $g(Y_T) = \varphi(\QV(Y_T))$, where $\QV(Y_T)$ denotes the quadratic variation of the path $Y_t$, see Appendix \ref{sec:stoch_int_quad_var} for the precise definition of this pathwise quadratic variation functional. It is straightforward to compute
\begin{align*}
\QV((Y_{t,\delta t})^h \otimes Z_{t+\delta t,T}^{y_t+h}) &= \QV(Y_{t-}) + (y_t - y_{t-})^2 + h^2 + \QV(Z_{t+\delta t,T}^{y_t+h}),\\
\QV((Y_t^h)_{t,\delta t} \otimes Z_{t+\delta t,T}^{y_t+h}) &= \QV(Y_{t-}) + (y_t - y_{t-} + h)^2 + \QV(Z_{t+\delta t,T}^{y_t+h}).
\end{align*}
Therefore,
\begin{align*}
\QV((Y_t^h)_{t,\delta t} \otimes Z_{t+\delta t,T}^{y_t+h}) - \QV((Y_{t,\delta t})^h \otimes Z_{t+\delta t,T}^{y_t+h}) = 2(y_t - y_{t-})h
\end{align*}
which implies $\phi \equiv 0$ in Equation (\ref{eq:g_condition}), for paths $Y_t$ without discontinuity at $t$. Hence, $f(Y_t) = \bE[\varphi(\QV(X_T)) \ | \ Y_t]$ is locally weakly path-dependent at continuous paths under any Markovian dynamics for $x$.

\end{example}

\section{Greeks for Path-Dependent Derivatives}\label{sec:greeks_sec}

\subsection{Introduction}\label{sec:intro}

In \cite{malliavin_greeks1}, the authors presented a computationally efficient way to calculate Greeks for some path-dependent derivatives using tools of the Malliavin calculus. More specifically, they considered a time-homogenous local volatility model,
\begin{align}
dx_t = r x_t dt + \sigma(x_t)dw_t, \label{eq:local_vol_homog}
\end{align}
and contracts of the form
\begin{align*}
g(Y_T) = \phi(y_{t_1}, \ldots, y_{t_n}),
\end{align*}
where $0 < t_1 < \cdots < t_n \leq T$ are fixed times and $\phi: \bR^n \longrightarrow \bR$ is such that $g(X_T) \in L^2(\Omega, \cF, \bP)$. Under these assumptions, it was shown that
$$\Delta_xf(Y_0) = \bE\left[\left. \phi(x_{t_1}, \ldots, x_{t_n}) \int_0^T  \frac{a(t)z_t}{\sigma(x_t)} dw_t \ \right| \ Y_0\right],$$
where $x$ is the solution of (\ref{eq:local_vol_homog}) with $x_0 = Y_0$, $z$ is the tangent process (or \textit{first variation process}) described by the SDE
\begin{align}
dz_t = r z_t dt + \sigma'(x_t) z_t dw_t \label{eq:tangent_process}
\end{align}
with $z_0=1$, and
\begin{align*}
a \in \Gamma = \left\{ a \in L^2[0,T] \ ; \ \int_0^{t_i} a(t) dt = 1, \ \forall \ i = 1,\ldots,n \right\}.
\end{align*}

It is also assumed that $\sigma$ is uniformly elliptic, which in the one-dimensional case boils down to $\sigma$ being bounded from below.

If we define the weight
\begin{align}
\pi = \int_0^T  \frac{a(t)z_t}{\sigma(x_t)} dw_t,\label{eq:pi}
\end{align}
which does not depend on the derivative contract $g$, we may restate the result above as:
$$\Delta_xf(Y_0) = \bE[ \phi(x_{t_1}, \ldots, x_{t_n}) \pi \ | \ Y_0].$$

\subsection{The Path-Dependent Volatility Model}

We would like to remind the reader that we are considering the more general path-dependent volatility models, see Section \ref{sec:pdv}. For arithmetic simplicity, we shall assume that $r=0$:
\begin{align}
dx_t = \sigma(X_t) dw_t. \label{eq:pd_vol_r0}
\end{align}
In this case of path-dependent volatility, we \textit{define} the tangent process $z$ to be the solution of the linear SDE:
\begin{align}
dz_t = \Delta_x \sigma(X_t) z_t dw_t, \label{eq:pd_tangent}
\end{align}
where $z_0 = 1$.

\begin{remark}

We would like to point it out that the proof that $z$ is actually the tangent process of $x$, meaning that $z_t = \partial_{x_0} x_t$, will not be pursued here. As it will be clear later, regarding our application, it is only important that the process $z$ cancels certain terms when we compute the differential $d(\Delta_x f(X_t) z_t)$. Besides, notice that, in the case of local volatility function, $z$ becomes the usual tangent process of $x$, i.e. $z_t = \partial_{x_0} x_t$.

\end{remark}

\begin{remark}

It is very important to notice that the dynamics of the underlying process, $x$, will obviously influence in the path-dependence of the price functional $f(Y_t) = \bE[g(X_T) \ | \ Y_t]$. In particular, the price of a derivative might be locally weakly path-dependent under a local volatility model, but strongly path-dependent when considering a path-dependent volatility model. This aspect of path-dependence is really intricate and hence, in the examples presented in this paper, we shall consider local volatility models. Nonetheless, the general results will be derived in the full generality that the functional It\^o calculus theory allows, i.e. under path-dependent volatility models.

\end{remark}

\begin{remark}\label{rmk:functional_z}

In the lines of what was shown in Appendix \ref{sec:stoch_int_quad_var}, we will consider the functional $z$ such that $z(X_t) = z_t$, i.e. $z(Y_t) = E(I_h(Y)_t)$, where $h(Y_t) = \frac{\Delta_x \sigma(Y_t)}{\sigma(Y_t)}$, see Appendix \ref{sec:stoch_int_quad_var} for the definition of the functionals $E$ and $I_h$. Following the arguments outlined in this appendix, one can easily show that $z$ satisfies
\begin{align}
\Delta_x z(Y_t) = \frac{\Delta_x \sigma(Y_{t-})}{\sigma(Y_{t-})} z(Y_{t-}), \ \Delta_{xx} z(Y_t) = 0 \mbox{ and } \Delta_t z(Y_t) = 0.\label{eq:func_derivatives_of_z}
\end{align}

\end{remark}

We now list the assumptions on $\sigma$ that will be used in what follows. They will be assumed to hold throughout the paper.

\begin{assumptions}[on the path-dependent volatility $\sigma$]\label{assump:sigma}

\begin{enumerate}

\item[]

\item $\sigma > 0$;

\item $\sigma \in \bC^{0,1}$, i.e. $\sigma$ is $\Lambda$-continuous, $\Delta_x \sigma$ exists and it is also $\Lambda$-continuous;

\item SDEs (\ref{eq:pd_vol_r0}) and (\ref{eq:pd_tangent}) have unique strong solutions, see Remark \ref{rmk:strong_solution};

\item the topological support of the process $x$ contains all the continuous functions in $[0,T]$ starting at $x_0$, see Remark \ref{rmk:stroock_varadhan}.

\end{enumerate}

\end{assumptions}


\begin{remark}[Strong Solutions]
Unique strong solution of (\ref{eq:pd_vol_r0}) is guaranteed by requiring that $\sigma$ satisfies the usual Lipschitz and linear growth conditions, cf. Remark \ref{rmk:strong_solution}. Strong solution of Equation (\ref{eq:pd_tangent}) follows from the continuity in $t$ of $\Delta_x \sigma(X_t)$ (this is true under the assumption that $\sigma \in \bC^{0,1}$).
\end{remark}

\begin{remark}
We are constraining ourselves to one-dimensional processes in order to make the exposition clearer, although the extension to multi-dimensional processes is straightforward. Moreover, the results in the following sections in this paper will be derived assuming smoothness in the sense of $\bC$, but one should expect that they could be generalized to consider smooth functional in the sense of $\cC$ as discussed in Appendix \ref{sec:stoch_int_quad_var}.
\end{remark}

\subsection{Greeks for Weakly Path-Dependent Functionals}\label{sec:delta_weakly}

\subsubsection{Delta}\label{sec:delta}

The Delta of a derivative contract is the sensitivity of its price with respect to the current value of the underlying asset. Hence, if $f(X_t)$ denotes the price of the aforesaid derivative at time $t$, its Delta is given by $\Delta_x f(X_t)$.

Consider a path-dependent derivative with maturity $T$ and contract $g:\Lambda_T \longrightarrow \bR$. The price of this derivative is given by the functional $f:\Lambda \longrightarrow \bR$:
$$f(Y_t) = \bE[g(X_T) \ | \ Y_t],$$
for any $Y_t \in \Lambda$. In what follows we will perform some formal computations and hence we assume $f$ as smooth as necessary for such calculations. By the Pricing PPDE, Theorem \ref{thm:functional_pde_theorem}, we know
$$\Delta_t f(Y_t) + \frac{1}{2} \sigma^2(Y_t) \Delta_{xx} f(Y_t) = 0,$$
for any continuous path $Y_t$. Now, consider the tangent process $z$ given by Equation (\ref{eq:pd_tangent}). The main idea is to apply the Functional It\^o Formula, Theorem \ref{thm:fif}, to $\Delta_x f(X_t) z_t$. First, notice that applying $\Delta_x$ to the PPDE gives
\begin{align}
\Delta_{xt} f(Y_t) &+ \sigma(Y_t) \Delta_x \sigma(Y_t) \Delta_{xx} f(Y_t) + \frac{1}{2} \sigma^2(Y_t) \Delta_{xxx} f(Y_t) = 0 \label{eq:delta_x_PDE}
\end{align}
In order to conclude the above, the following result is needed: \textit{if $f(Y_t) = 0$, for all continuous paths Y, and $f \in \bC^{1,1}$ , then $\Delta_x f(Y_t)  = 0$, for all continuous paths Y}. The proof of this can be found in \cite[Theorem 2.2]{fournie_cont_thesis}. Hence
\begin{align*}
&d(\Delta_x f(X_t) z_t) = z_t d(\Delta_x f(X_t)) + \Delta_x f(X_t) dz_t + d(\Delta_x f(X_t)) dz_t\\
&=\left(\Delta_{tx} f(X_t) + \frac{1}{2} \sigma^2(X_t) \Delta_{xxx} f(X_t) + \sigma(X_t) \Delta_x \sigma(X_t) \Delta_{xx} f(X_t)\right)z_t dt \\
& + \left(\Delta_x \sigma(X_t) \Delta_x f(X_t) + \sigma(X_t) \Delta_{xx} f(X_t) \right) z_t dw_t.
\end{align*}
Moreover, we define the local martingale
\begin{align}
m_t = \int_0^t \left(\frac{\Delta_x \sigma(X_s)}{\sigma(X_s)} \Delta_x f(X_s) + \Delta_{xx} f(X_s) \right) z_s dx_s, \label{eq:martingale}
\end{align}
with $m_0 = 0$, where we are assuming certain integrability condition of the integrand. Using Equation (\ref{eq:delta_x_PDE}), we are able to derive the formula
\begin{align}
d(\Delta_x f(X_t) z_t) &= (\Delta_{tx} f(X_t) - \Delta_{xt} f(X_t))z_t dt + dm_t = -\fL f(X_t)z_tdt + dm_t. \label{eq:lie_bracket_delta_xf}
\end{align}

We start by stating the assumptions on the functional $f$:

%
%
%
%
%
%
%

\begin{assumptions}[on the regularity of the price functional $f$]\label{assump:delta_reg}

\begin{enumerate}

\item[]

\item the Lie bracket of $f$, $\fL f$, exists;

\item $f \in \bC^{2,3}$;

\item $g(X_T) \in L^2(\Omega, \cF, \bP)$.


\end{enumerate}

\end{assumptions}

\begin{assumptions}\label{assump:delta_path}
$\fL f(Y_t) = 0$, for continuous paths $Y_t$.
\end{assumptions}

In particular if $f$ is locally weakly path-dependent, then $f$ satisfies Assumptions \ref{assump:delta_path}. Hence, the following result holds true:

\begin{theorem}\label{thm:delta_0}
Consider a path-dependent derivative with maturity $T$ and contract $g:\Lambda_T \longrightarrow \bR$. If the price of this derivative, denoted by $f$, satisfies Assumptions \ref{assump:delta_reg} and \ref{assump:delta_path}, then $(\Delta_x f(X_t) z_t)_{t \in [0,T]}$ is a local martingale and the following formula for the Delta is valid
$$\Delta_x f(Y_0) = \bE\left[\left.g(X_T) \frac{1}{T}  \int_0^T \frac{z_t}{\sigma(X_t)} dw_t \ \right| Y_0 \ \right].$$
\end{theorem}

\begin{proof}
By a localization argument outlined in Appendix \ref{sec:proof_delta_0}, we may assume that $f \in \cD_x$ and that $x$ and $m$ are martingales.\\

From Equation (\ref{eq:lie_bracket_delta_xf}), Assumption (\ref{assump:delta_path}) and since $X_t$ is a continuous path $\bP$-almost surely, we conclude
$$\Delta_x f(X_t)z_t = \Delta_x f(X_0) + m_t,$$
and then $(\Delta_x f(X_t) z_t)_{t \in [0,T]}$ is clearly a martingale. Now, integrating with respect to $t$, we get
$$\int_0^T \Delta_x f(X_t) z_t dt = \Delta_x f(X_0) T + \int_0^T m_t  dt.$$
Then taking expectations and noticing $\bE[m_t] = m_0 = 0$, we get
$$\bE\left[\int_0^T \Delta_x f(X_t) z_t dt\right] = \Delta_x f(X_0) T,$$
which implies
\begin{align}
\Delta_x f(X_0) &= \bE\left[\int_0^T \Delta_x f(X_t) \frac{1}{T} \frac{z_t}{\sigma^2(X_t)} \sigma^2(X_t) dt\right] \\
&= \left\langle \Delta_x f(X), \frac{1}{T} \frac{z}{\sigma^2(X)} \right\rangle_{\cH^2_x}.
\end{align}
Finally, since $f(X)$ and $x$ are martingales, by the integration by parts formula (\ref{eq:ibp}),
\begin{align*}
\Delta_x f(X_0) &= \left\langle f(X), \cI_x\left(\frac{1}{T} \frac{z}{\sigma^2(X)}\right) \right\rangle_{\cM^2_x} = \bE\left[g(X_T) \frac{1}{T}  \int_0^T \frac{z_t}{\sigma(X_t)} dw_t\right].
\end{align*}
\end{proof}

\begin{remark}\label{rmk:delta_0_bs}
In the Black--Scholes model (i.e. $\sigma(Y_t)= \sigma y_t$), we find the same result as in \cite{malliavin_greeks1}
$$ \Delta_{x}f(X_0)=\bE\left[g(x_T) \frac{w_T}{x_0 \sigma T}\right].$$
\end{remark}

\begin{remark}
Theorem \ref{thm:delta_0} states that, for locally weakly path-dependent functionals, the weight can take the form
\begin{align}
\pi = \frac{1}{T} \int_0^T  \frac{z_t}{\sigma(X_t)} dw_t,
\end{align}
cf. Equation (\ref{eq:pi}).
\end{remark}

\begin{remark}

Theorem \ref{thm:delta_0} also enlightens the question when the Delta is a martingale. The theorem affirms that the lost of martingality of the Delta comes from two factors: the stock price model through its tangent process $z$ and the path-dependence of the derivative contract in question.

For instance, let us consider a call option. It is a well-know fact that, under the Black-Scholes model, the Delta is not a martingale. Although the price of a call option is locally weakly path-dependent (actually it is path-independent), the tangent process in this model is given by $z_t = x_t/x_0$. On the other hand, under the Bachelier model, the Delta of a locally weakly path-dependent derivative contract is indeed a martingale, since $z_t = 1$ in this case.

\end{remark}

\begin{remark}

One would expect that the assumption $f \in \bC^{2,3}$ could be removed by using a density argument. However, there are no results in this direction available at the current development of the functional It\^o calculus theory and to develop such density arguments is outside the scope of this paper.

\end{remark}

\begin{corollary}\label{cor:delta_s}
Under the same hypotheses as in Theorem \ref{thm:delta_0}, for any $s \in [0,T]$, one has
\begin{align}
\Delta_x f(Y_s) = \frac{1}{(T-s) z(Y_s)} \bE\left[ \left. g(X_T) \int_s^T \frac{z_t}{\sigma(X_t)} dw_t \ \right| \ Y_s \right], \label{eq:delta_s}
\end{align}
where $z(Y_s)$ is the functional version of the tangent process $z$, see Remark \ref{rmk:functional_z}.
\end{corollary}

\begin{proof}
The same argument is applied with some minor differences. Notice the study of the integration by parts formula for $\Delta_x$ can be easily extended to handle the conditional expectation.
\end{proof}

\subsubsection{Strongly Path-Dependent Functionals}

How would these formulas change if $f$ were strongly path-dependent? The integral form of Equation (\ref{eq:lie_bracket_delta_xf}) is
\begin{align}
\Delta_x f(X_0) = \Delta_x f(X_t)z_t + \int_0^t \fL f(X_s) z_sds - m_t. \label{eq:mart_non_zero_lie}
\end{align}
Integrating with respect to $t$ and taking expectation, we get
\begin{align}\label{eq:delta_path}
\Delta_x f(X_0) = \bE\left[\frac{1}{T}\int_0^T \Delta_x f(X_t) z_t dt \right] + \bE\left[\frac{1}{T}\int_0^T \int_0^t \fL f(X_s) z_s ds dt \right].
\end{align}
Now, for the first expectation, we use the same argument as in Theorem \ref{thm:delta_0} to conclude
\begin{align}
\bE\left[\frac{1}{T}\int_0^T \Delta_x f(X_t) z_t dt \right] = \bE\left[g(X_T) \frac{1}{T}\int_0^T \frac{z_t}{\sigma(X_t)} dw_t \right].\label{eq:delta_path_first}
\end{align}


We hence proved the following theorem:

\begin{theorem}\label{thm:delta_strong}
For a path-dependent derivative with maturity $T$ and contract $g$ such that its price, denoted by the functional $f$, satisfies Assumptions \ref{assump:delta_reg}, the following formula for the Delta holds:
\begin{align}\label{eq:delta_path_2}
\Delta_x f(X_0) = \bE\left[g(X_T) \frac{1}{T}\int_0^T \frac{z_t}{\sigma(X_t)} dw_t \right] + \bE\left[\frac{1}{T}\int_0^T \int_0^t \fL f(X_s) z_s ds dt \right].
\end{align}
\end{theorem}

Since the formula above makes reference to $f$ and its Lie bracket, it is not as computationally appealing as the formula derived for locally weakly path-dependent functionals, see Theorem \ref{thm:delta_0}. To achieve better results computational-wise, for the second term of the right-hand side of (\ref{eq:delta_path_2}), future research should focus on the adjoint and/or an integration by parts for $\Delta_t$ and $\Delta_x$ in $\cH_x^2$. An integration by parts formula for $\Delta_x$ in $\cH_x^2$ is presented in \cite[Section 3]{fito_dupire}.

In any event, an important interpretation of the second term of the right-hand side of Equation (\ref{eq:delta_path_2}) is as a \textit{path-dependent correction} to the locally weakly path-dependent ``Delta" of Equation (\ref{eq:delta_path_first}), which does not take into consideration the strong path-dependence structure of the derivative contract. This is one of the most important achievements of the functional It\^o calculus framework: it allows us to quantify how the path-dependence of the functional influences the Delta of this contract. We would like to call attention to the fact that this was not achieved within the Malliavin calculus framework.

In the next sections we provide formulas for the Gamma and the Vega of locally weakly path-dependent derivative contracts. Similar formulas and proofs for the different classifications of path-dependence of Definition \ref{def:classif} can be derived following akin arguments.

\subsubsection{Gamma}

The Gamma of a derivative is the sensitivity of its Delta with respect to the current value of the underlying asset, i.e. $\Delta_{xx} f(X_t)$. Here we will derive a similar formula to (\ref{eq:delta_s}) for the Gamma.


\begin{assumptions}\label{assump:sigma_gamma}
$\Delta_t \sigma = \Delta_{tx} \sigma = 0$ in $\Lambda$.


\end{assumptions}

Notice that Assumption \ref{assump:sigma_gamma} is satisfied for time-homogenous local volatility models, see Equation (\ref{eq:local_vol_homog}).

\begin{theorem}\label{thm:gamma}
Under Assumptions \ref{assump:delta_reg} and \ref{assump:delta_path} for $f$ and $\Delta_x f$ and additionally assuming that $\sigma$ satisfies Assumptions \ref{assump:sigma_gamma}, we find
$$\Delta_{xx} f(X_s) = \bE[ g(X_T) \xi_{s,T} \ | \ X_s],$$
where
\begin{align}
&\eta_s = \int_0^s \frac{z_t}{\sigma(X_t)} dw_t, \\
&\xi_{s,T} = \frac{(\eta_T - \eta_s)^2}{(T-s)^2z_s^2} - \frac{\Delta_x \sigma(X_s)}{\sigma(X_s)} \frac{\eta_T - \eta_s}{(T - s) z_s}  - \frac{1}{(T - s)\sigma^2(X_s)}. \label{eq:xi_sT}
\end{align}
\end{theorem}

\begin{proof}

Firstly, there exist functionals $z$ and $\eta$ such that $z(X_t) = z_t$ and $\eta(X_t) = \eta_t$ a.s. By the functional derivatives formulas shown in Appendix \ref{sec:stoch_int_quad_var}, we have derived the functional derivatives of $z$ in Equation (\ref{eq:func_derivatives_of_z}). Additionally, by the same arguments one can easily conclude that
\begin{align*}
\Delta_x \eta(Y_t) = \frac{z(Y_{t-})}{\sigma^2(Y_{t-})}, \ \Delta_{xx} \eta(Y_t) = 0 \mbox{ and } \Delta_t \eta(Y_t) = 0.
\end{align*}

Remember now the following formula given in Corollary \ref{cor:delta_s}:
$$(T - s) z(Y_s) \Delta_x f(Y_s) + f(Y_s) \eta(Y_s) = \bE\left[\left. g(X_T) \eta(X_T) \right| Y_s \right].$$
Define then $\tilde{g}(Y_T) = g(Y_T) \eta(Y_T)$ and $\tilde{f}(Y_s) = \bE[\tilde{g}(X_T) \ | \ Y_s]$. Hence,
$$\tilde{f}(Y_s) = (T - s) z(Y_s) \Delta_x f(Y_s) + f(Y_s) \eta(Y_s).$$
It is easy to see that $\tilde{f}$ satisfies Assumptions \ref{assump:delta_reg}, since $\Delta_x f$ and $f$ satisfy this assumption themselves. Now, in order to apply the same argument as in the proof of the Theorem \ref{thm:delta_0}, it is necessary to prove $\fL \tilde{f} = 0$:
\begin{align}
&\Delta_x \tilde{f}(Y_s) = (T - s) \frac{\Delta_x \sigma(Y_{s-})}{\sigma(Y_{s-})}z(Y_{s-}) \Delta_x f(Y_s) + (T - s)z(Y_s) \Delta_{xx} f(Y_s) \label{eq:delta_x_f_tilde}\\
& \hskip 1.5cm + \Delta_x f(Y_s) \eta(Y_s) +  f(Y_s) \frac{z(Y_{s-})}{\sigma^2(Y_{s-})}, \nonumber\\
&\Delta_t \tilde{f}(Y_s) = -z(Y_s) \Delta_x f(Y_s) + (T - s) z(Y_s) \Delta_{tx} f(Y_s) + \Delta_t f(Y_s) \eta(Y_s).\label{eq:delta_t_f_tilde}
\end{align}
Let us now compute the mixed derivatives. For this, we have to assume that $y_{s-} = y_s$, which implies $Y_{s-} = Y_s$, see Equation (\ref{eq:path_x_t-}). In particular, the following computations work when $Y_s$ is continuous.
\begin{align*}
&\Delta_{tx} \tilde{f}(Y_s) = - \frac{\Delta_x \sigma(Y_s)}{\sigma(Y_s)}z(Y_s) \Delta_x f(Y_s) + (T - s) \frac{\Delta_{tx} \sigma(Y_s)}{\sigma(Y_s)}z(Y_s) \Delta_{x} f(Y_s) \\
& \hskip 1.5cm - (T - s) \frac{\Delta_{x} \sigma(Y_s)}{\sigma^2(Y_s)}\Delta_t \sigma(Y_s) z(Y_s) \Delta_{x} f(Y_s)\\
& \hskip 1.5cm + (T - s) \frac{\Delta_x \sigma(Y_s)}{\sigma(Y_s)}z(Y_s) \Delta_{tx} f(Y_s) - z(Y_s) \Delta_{xx} f(Y_s) \\
& \hskip 1.5cm + (T - s)z(Y_s) \Delta_{txx} f(Y_s) + \Delta_{tx} f(Y_s) \eta(Y_s) + \Delta_t f(Y_s) \frac{z(Y_s)}{\sigma^2(Y_s)}, \\
& \hskip 1.5cm - 2 f(Y_s)z(Y_s) \frac{\Delta_t \sigma(Y_t)}{\sigma^3(Y_t)}, \\ \\
&\Delta_{xt} \tilde{f}(Y_s) = - \frac{\Delta_x \sigma(Y_s)}{\sigma(Y_s)}z(Y_s) \Delta_x f(Y_s)  - z(Y_s) \Delta_{xx} f(Y_s) + (T - s) \frac{\Delta_x \sigma(Y_s)}{\sigma(Y_s)}z(Y_s) \Delta_{tx} f(Y_s) \\
& \hskip 1.5cm + (T - s)z(Y_s) \Delta_{xtx} f(Y_s) + \Delta_{xt} f(Y_s) \eta(Y_s) + \Delta_t f(Y_s) \frac{z(Y_s)}{\sigma^2(Y_s)}.
\end{align*}
Finally, since $\fL f(Y_s) = 0 = \fL (\Delta_x f)(Y_s)$, for continuous paths $Y_s$, and Assumption \ref{assump:sigma_gamma} is true, we find $\fL \tilde{f}(Y_s) = 0$, for continuous paths $Y_s$. Hence, $\tilde{f}$ satisfies Assumptions \ref{assump:delta_reg} and \ref{assump:delta_path}, and then, by Theorem \ref{thm:delta_0}, $(\Delta_x \tilde{f}(X_s) z_s)_{s \in [0,T]}$ is a martingale. Therefore
$$(T - s) z_s \Delta_x \tilde{f}(X_s) + \tilde{f}(X_s) \int_0^s \frac{z_t}{\sigma(X_t)} dw_t = \bE\left[\left. \tilde{g}(X_T) \int_0^T \frac{z_t}{\sigma(X_t)}dw_t \ \right| \ X_s \right].$$
By Equation (\ref{eq:delta_x_f_tilde}), we find
\begin{align*}
\Delta_x \tilde{f}(X_s) &= (T - s) z_s \frac{\Delta_x \sigma(X_s)}{\sigma(X_s)} \Delta_x f(X_s) + (T - s) z_s \Delta_{xx} f(X_s) \\
&+ \Delta_x f(X_s) \int_0^s \frac{z_t}{\sigma(X_t)} dw_t + f(X_s) \frac{z_s}{\sigma^2(X_s)}.
\end{align*}
Lastly, the result can be easily derived from the equation above.
\end{proof}

\begin{corollary}
At $s = 0$,
$$\Delta_{xx} f(X_0) = \bE[ g(X_T) \xi],$$
where
$$\xi = \xi_{0,T} = \pi^2 - \frac{\Delta_x \sigma(X_0)}{\sigma(X_0)} \pi  - \frac{1}{T\sigma^2(X_0)},$$
since $\pi = \eta_T/T$.
\end{corollary}

\begin{remark}
In the Black--Scholes model, we find the same result as in \cite{malliavin_greeks1}:
$$ \Delta_{xx}f(X_0)=\bE\left[g(X_T)\frac{1}{X_0^2 \sigma T}\left(\frac{w_T^2}{\sigma T}-w_T-\frac{1}{\sigma}\right)\right].$$
However, in \cite{malliavin_greeks1} the Gamma was derived only under the Black--Scholes model and for path-independent derivatives with contract of the form $g(X_T) = \phi(x_T)$.
\end{remark}

\subsubsection{Vega}\label{sec:vega}

In this section, we restrict ourselves to time-homogeneous local volatility models, i.e. $\sigma(Y_t) = \sigma(y_t)$. Consistently to \cite{fito_dupire}, we define the Vega of $f(X_t)$ as the Fr\'echet derivative of $f(X_t)$ with respect to $v = \sigma^2$. Using the result presented in \cite[Section 4, Example 1]{fito_dupire}, we know that the Vega of $f(X_t)$ in the direction of $u$ is given by
\begin{align}
\langle \nabla_v f, u \rangle = \lim_{\eps \to 0} \frac{\bE_{v_0 + \eps u}[g(X_T)] - \bE_{v_0}[g(X_T)]}{\eps} = \int_0^T \int_{\bR} u(t,x) m(t,x) dxdt, \label{eq:vega_def}
\end{align}
where
$$m(t,x) = \frac{1}{2}\bE_{v_0}\left[\Delta_{xx} f(X_t)  \ | \ x_t = x\right] p^{v_0}(t,x).$$
Here, $\bE_{v_0}$ is the expectation under the local volatility model (\ref{eq:local_vol_homog}) with $v_0 = \sigma^2$ and $p^{v_0}(t,x)$ is the density of $x_t$ under $v_0$.

\begin{theorem}\label{thm:vega}
Under the hypotheses of Theorem \ref{thm:gamma}, the Vega satisfies
$$\langle \nabla_v f, u \rangle = \bE_{v_0}\left[g(X_T) \ \frac{1}{2} \int_0^T  u(t,x_t) \xi_{t,T} dt \right].$$
where $\xi_{t,T}$ is given by Equation (\ref{eq:xi_sT}). Moreover,
\begin{align}
m(t,x) = \frac{1}{2}\bE_{v_0}\left[g(X_T) \xi_{t,T}  \ | \ x_t = x\right] p^{v_0}(t,x). \label{eq:vega_m}
\end{align}
\end{theorem}

\begin{proof}
Equation (\ref{eq:vega_def}) can be rewritten as
$$\langle \nabla_v f, u \rangle = \frac{1}{2} \int_0^T \bE_{v_0}[u(t,x_t) \Delta_{xx} f(X_t)] dt.$$
Assuming the conditions of Theorem \ref{thm:gamma} are satisfied, then
$$\Delta_{xx} f(X_t) = \bE_{v_0}[ g(X_T) \xi_{t,T} \ | \  X_t],$$
and thus the following is true:
\begin{align}
\langle \nabla_v f, u \rangle &= \frac{1}{2} \int_0^T \bE_{v_0}[u(t,x_t) \bE_{v_0}[ g(X_T) \xi_{t,T} \ | \  X_t]] dt \nonumber \\
&= \frac{1}{2} \int_0^T \bE_{v_0}[ u(t,x_t) g(X_T) \xi_{t,T} ] dt \nonumber \\
&= \bE_{v_0}\left[g(X_T) \ \frac{1}{2} \int_0^T  u(t,x_t) \xi_{t,T} dt \right]. \label{eq:vega}
\end{align}
Notice now that
$$\bE_{v_0}\left[\Delta_{xx} f(X_t)  \ | \ x_t\right] = \bE_{v_0}\left[\bE_{v_0}[ g(X_T) \xi_{t,T} \ | \  X_t]  \ | \ x_t\right] = \bE_{v_0}\left[g(X_T) \xi_{t,T}  \ | \ x_t\right],$$
which implies (\ref{eq:vega_m}).
\end{proof}

\begin{remark}\label{rmk:bridges}
The results presented in the previous Theorem allows us to more efficiently compute the Vega of a path-dependent derivative in a local volatility model, namely we avoid the computation the functional second derivative of the price functional, $\Delta_{xx} f$.

The expectation in Equation (\ref{eq:vega_m}) should be understood as follows: the process starts at $(0,x_0)$ and it is simulated up to time $T$, but with the condition that $x_t = x$ (this is a spot value condition). In order to do this, one needs to simulate diffusion bridges, i.e. a diffusion under the condition that it starts at $(0,x_0)$ and pass at $(t,x_t)$. This type of conditional expectation does not appear in the case for Deltas and Gammas, since the conditional expectation involved in their computations are conditioned to the full path $X_t$.
\end{remark}

\begin{remark}
Comparing this result with the one presented in \cite{malliavin_greeks1}, we notice that Dupire's formula (\ref{eq:vega_def}) avoids the necessity to compute Skorohod integrals. Actually, one can show that the formula for the Vega in \cite{malliavin_greeks1} can be simplified to (\ref{eq:vega_def}) when $g(X_T) = \phi(x_T)$.
\end{remark}

\subsubsection{Numerical Example}\label{ex:quadratic_variation}

Volatility derivatives are financial contracts such that their underlying asset is a measurement of volatility or variance, such as the realized volatility over a pre-determined period or the Chicago Board Options Exchange Market Volatility Index (VIX).

In this example, we will consider the continuous-time version of options on realized variance, more precisely \textit{options on quadratic variation}, see for instance \cite{lee_quad_var_derivative}. This example was not dealt in the Malliavin calculus setting.

We will consider a payoff functional $g$ of the form $g(Y_T) = \phi(y_T, QV(Y_T))$, where $QV$ is the functional representing the pathwise quadratic variation of the price path, we refer the reader to Appendix \ref{sec:stoch_int_quad_var}. Particularly, we will examine a Call option with a variance European knock-out barrier, i.e. $\phi(y, QV) = (y - K)^+1_{\{QV < H\}}$. This derivative is called a VKO Call option; it is a commonly traded exotic derivative in the Foreign Exchange markets.

The price functional $f(Y_t) = \bE[g(X_T) \ | \ Y_t]$ is defined as in Equation (\ref{eq:conditioned_expec}). We start by observing that, under a local volatility model, an augmentation-of-variable argument shows that one can write $f(Y_t) = \psi(t, y_t, QV(Y_t))$. Following this characterization, one could prove the smoothness of the function $\psi$ (and hence of the functional $f$) using classical tools of PDE. Hence, $f$ satisfies Assumptions \ref{assump:delta_reg}.

To analyze the path-dependence of this derivative, we would like to derive the Lie bracket of $f$. Unfortunately, the time functional derivative of $\Delta_x QV$ does not exist in the whole $\Lambda$. Nonetheless, we are able to conclude that $\fL QV(Y_t) = 0$, for continuous paths $Y_t$, see Appendix \ref{sec:stoch_int_quad_var}. Hence, under a local volatility model, the same holds for $f$, by Proposition \ref{prop:lie_bracket_chain_rule}. Therefore, $f$ satisfies Assumptions \ref{assump:delta_path}. See Appendix \ref{sec:stoch_int_quad_var} for additional details.

To make this example computationally interesting for the calculation of the Delta and Gamma, we assume that $x$ follows a CEV (Constant Elasticity of Variance) model, i.e. $\sigma(x) = \max\{\sigma x^{\gamma}, \alpha\}$, where $\alpha=0.001$ is a lower bound in order to ensure that $\sigma$ is bounded from below, see Remark \ref{rmk:stroock_varadhan}.

In the specific case of the computation of the Vega, we will assume the Black--Scholes model (i.e. we take $\gamma = 1$). More complex local volatility models, like the CEV itself, could be considered. However, it would be computationally challenging to simulate its diffusion bridges (see Remark \ref{rmk:bridges}) and hence outside the scope of this paper. The order of magnitude of the quadratic variation of $x$ with $\gamma = 1$ and $\gamma = 0.5$ are different. We take that into account when computing the Vega.

Below, in Figure \ref{fig:mc_qv}, we show the convergence plots of $\Delta_x f(X_0)$ and $\Delta_{xx} f(X_0)$. These quantities are computed using Theorems \ref{thm:delta_0} and \ref{thm:gamma}. Moreover, in Figure \ref{fig:mc_qv_fd}, we compare the convergence plots using the weights given by Theorems \ref{thm:delta_0} and \ref{thm:gamma} with the standard Finite Difference estimation  (i.e. $\Delta_x f(X_0) = (f(X_0^{h}) - f(X_0^{-h})/(2h))$ and $\Delta_x f(X_0) = (f(X_0^{h}) -2f(X_0) + f(X_0^{-h})/h^2)$). Finally, we present the plot of the Vega of $f$ as defined in Section \ref{sec:vega}. More precisely, we plot $m(x,t)$ computed by Equation \ref{eq:vega_m}.

Considering the parameters given in Table \ref{tab:ParametersQV}, we show the results in Table \ref{tab:PriceDeltaQV} and in Figures \ref{fig:mc_qv}, \ref{fig:mc_qv_fd} and \ref{fig:vega_qv}.

The reader should notice that Finite Difference estimation for the Delta and Gamma perform very poorly. The reason is the discontinuity of the knock-out volatility barrier. This feature adds significant noise to the problem, as one can easily see in Figure \ref{fig:mc_qv_fd}.
\begin{table}[h!]
\begin{center}
\begin{tabular}{l c}
\hline Parameter & Value  \\
\hline
Initial Value $(X_0)$ & 100 \\
Volatility $(\sigma)$ & 0.2 \\
CEV Parameter $(\gamma)$ & 0.5/1.0 \\
Strike $(K)$ & 100 \\
Variance Barrier $(H)$ & 4.0 \\
Maturity $(T)$ & 1.0 \\
\hline
\end{tabular}
\caption{Parameters of the example on the VKO Call option.} \label{tab:ParametersQV}
\vspace{0.3cm}
\begin{tabular}{c c c}
\hline
\empty & Mean & Standard Error \\
\hline
$f(X_0)$ & 0.3624 & 0.0026\\
$\Delta_x f(X_0)$ & 0.2087 & 0.0022\\
$\Delta_{xx} f(X_0)$ & 0.0604 &  0.0018\\
$\Delta_x f(X_0)$ (FD) & 0.3683 & 0.1884\\
$\Delta_{xx} f(X_0)$ (FD) & -28.8123 &  37.5483\\
\hline
\end{tabular}
\caption{Monte Carlo Estimation of the Price, Delta and Gamma of the VKO Call option.} \label{tab:PriceDeltaQV}
\end{center}
\end{table}

\begin{figure}[h!]
\centering
\centerline{\includegraphics[scale = 0.5]{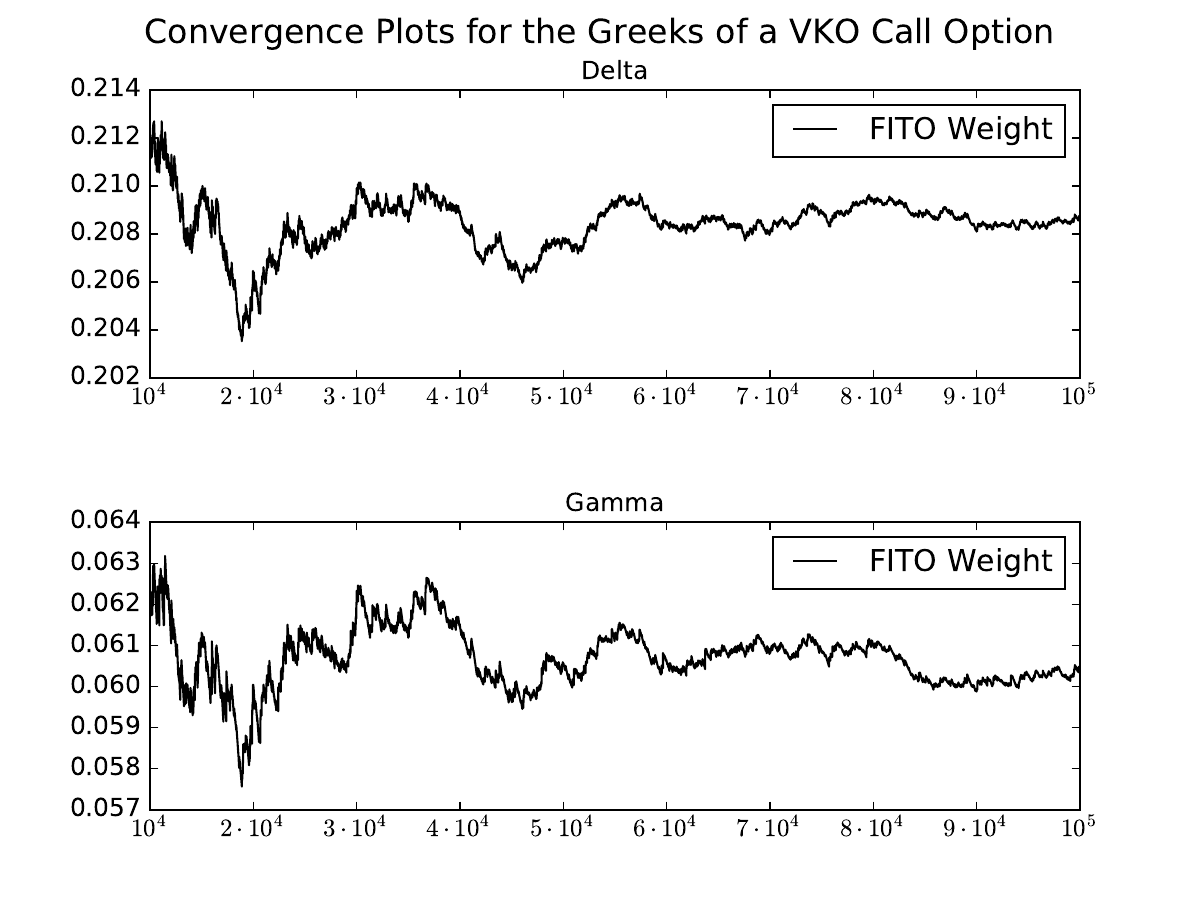}}
\renewcommand\figurename{Figure}
\caption{Convergence Plot of the Monte Carlo Method to Compute $\Delta_x f(X_0)$ and $\Delta_{xx} f(X_0)$.}
\label{fig:mc_qv}
\centerline{\includegraphics[scale = 0.5]{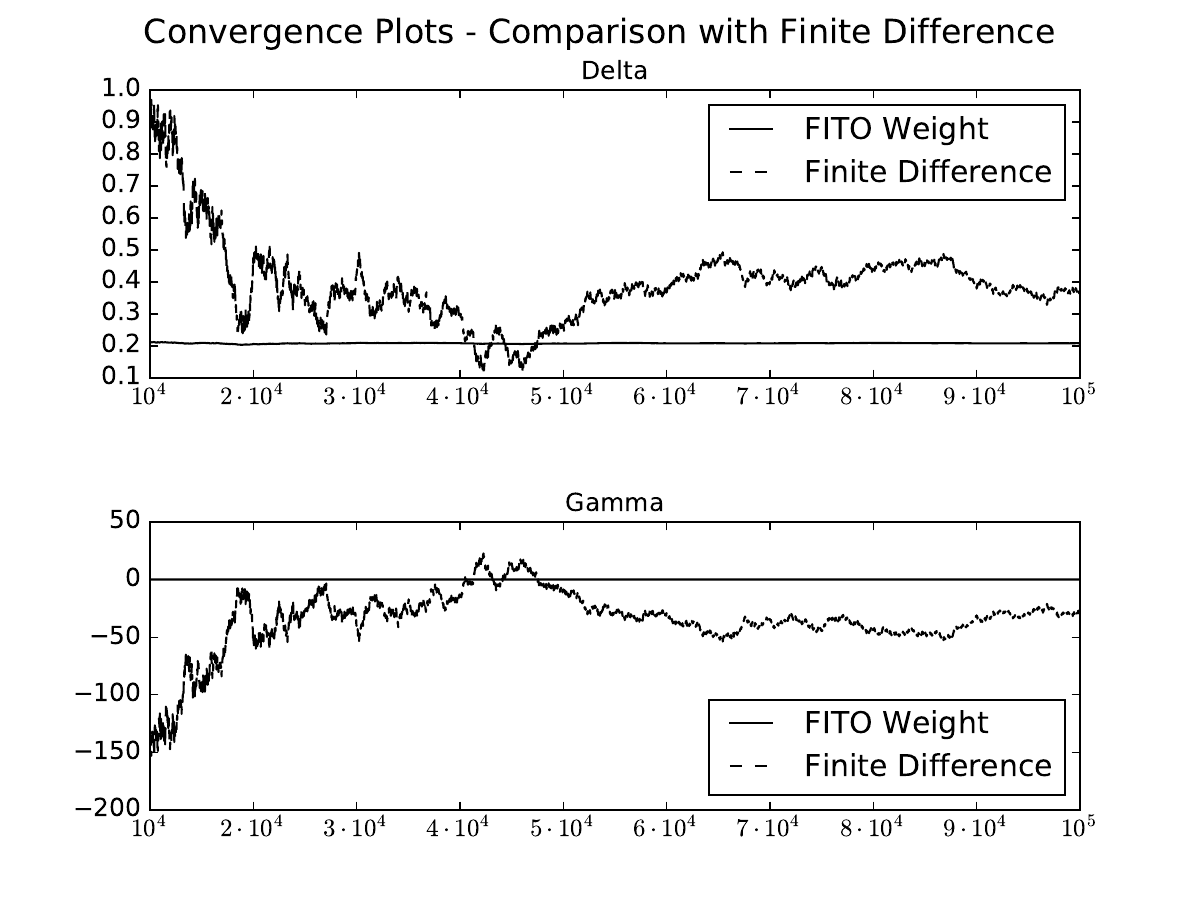}}
\vspace{-0.6cm}
\caption{Convergence Plot - Comparison with Finite Difference Method.}
\label{fig:mc_qv_fd}
\end{figure}

\begin{figure}[h!]
\centering
\centerline{\includegraphics[scale = 0.4]{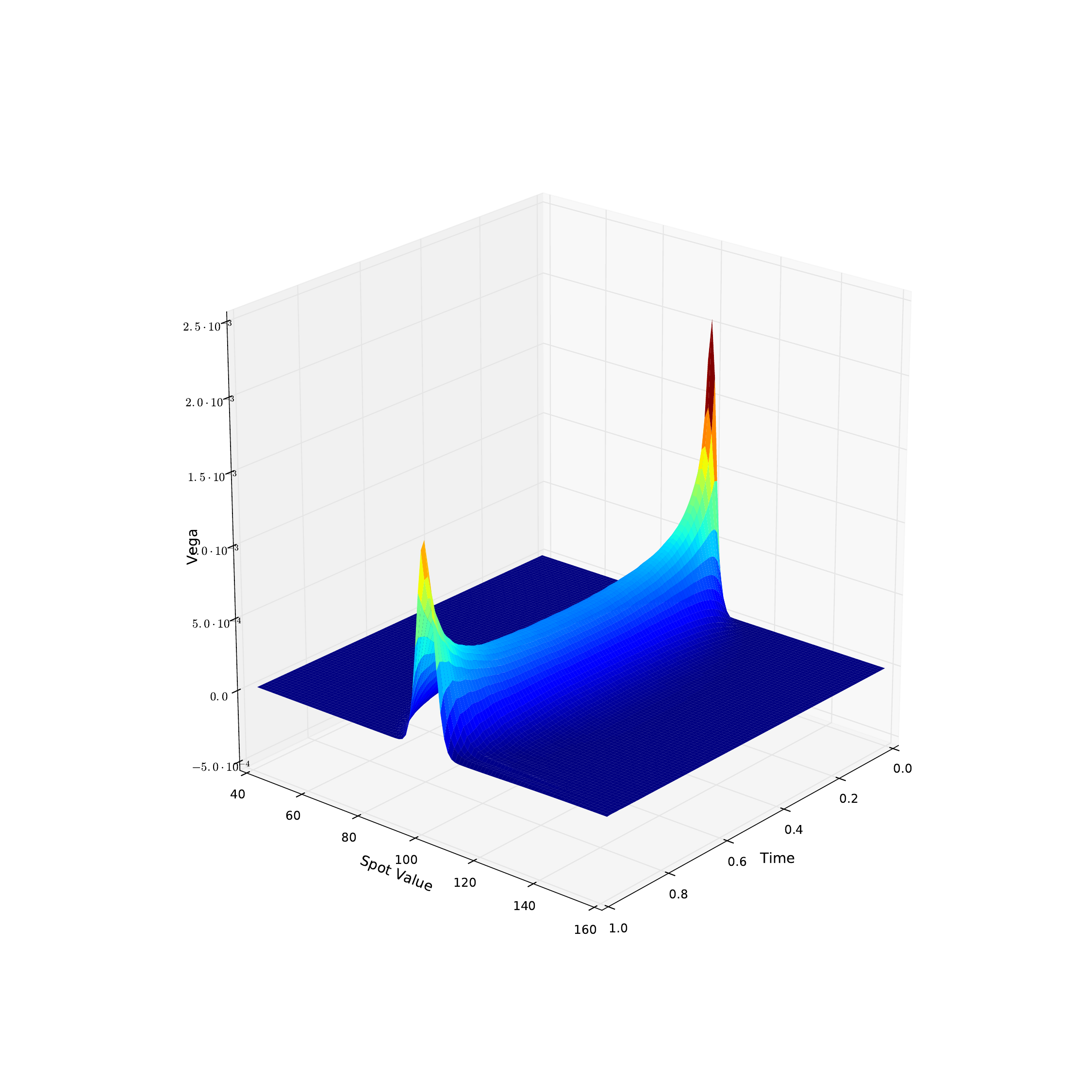}}
\vspace{-2cm}
\caption{Plot of $m(x,t)$ - the Vega for a VKO Call option. The axes Time and Spot Value mean $t$ and $x$, respectively.}
\label{fig:vega_qv}
\end{figure}

\clearpage

\subsection{More on Delta}

In the following, we will derive formulas for the Delta of a derivative contract distinguishing each path-dependence structure presented in Definition \ref{def:classif}.

The goal of this section is twofold: show how the result found in \cite{malliavin_greeks1} using Malliavin calculus can be achieved using functional It\^o calculus and then provide a better understanding of the assumption used in the Malliavin calculus framework that the contract $g$ is of the form:
\begin{align}\label{eq:contract_finite}
g(Y_T) = \phi(y_{t_1}, \ldots, y_{t_n}).
\end{align}
In short, this assumption implies that contracts of this form generate derivatives prices that are discretely monitored functionals, see Definition \ref{def:classif}. The main feature of these functionals is that they exhibit local weak path-dependence but for the finite set of times $\{t_1,\ldots,t_n\}$, see Proposition \ref{prop:lie_bracket_discret}.

\subsubsection{Discretely Monitored Functionals}

In this section, we consider a simple modification of the method described in Section \ref{sec:delta_weakly} to handle discretely monitored functionals as studied in \cite{malliavin_greeks1}, see Equation (\ref{eq:contract_finite}).

\begin{theorem}\label{thm:delta_discret}
Assume the no-arbitrage price of a path-dependent derivative, denoted by $f$, is a discretely monitored functional and that $f$ satisfies Assumptions \ref{assump:delta_reg}. Hence, we find the same representation for the Delta as in \cite{malliavin_greeks1}:
\begin{align}
\Delta_x f(X_0) = \bE\left[g(X_T) \int_0^T \frac{a(t)z_t}{\sigma(X_t)} dw_t\right],
\end{align}
for any $a \in \Gamma$, where
\begin{align}
\Gamma = \left\{ a \in L^2[0,T] \ ; \ \int_0^{t_i} a(t) dt = 1, \ \forall \ i = 1,\ldots,n \right\}. \label{eq:gamma_fournie}
\end{align}
\end{theorem}

\begin{proof}

To focus on the essential arguments of the proof, we consider the case with only two monitoring dates $t_1 < T$. This setting allow us to introduce all the elements of the proof without the burden of heavy notations. A similar reasoning could be applied to the general case.

As we have seen in Equation (\ref{eq:lie_bracket_delta_xf}),
\begin{align}\label{eq:sde_proof_delta_discret}
d(\Delta_x f(X_t) z_t) = -\fL f(X_t)z_tdt + dm_t,
\end{align}
with $(m_t)_{t \in [0,T]}$ being a local martingale. By well-known localization arguments (see Appendix \ref{sec:proof_delta_0}), we assume that $x$ and $m$ are martingales and that $f \in \cD_x$. As seen in Proposition \ref{prop:lie_bracket_discret}, we have  $\fL f(X_t)=0$  for all $t \in [0,t_1) \cup (t_1,T]$. Since $\fL f$ does not exist at $t = t_1$, we are only able integrate Equation (\ref{eq:sde_proof_delta_discret}) over intervals not containing $t_1$. Fix $\eps > 0$ and for $t \in (t_1,T]$, we integrate the SDE (\ref{eq:sde_proof_delta_discret}) over the interval $[t_1+\eps,t]$, we get
$$\Delta_x f(X_t) z_t = \Delta_x f(X_{t_1+\eps}) z_{t_1+\eps}+ m_{t}-m_{t_1+\eps}.$$
So, multiplying by any $a \in \Gamma$ and integrating with respect to $t$ from $t_1+\eps$ to $T$, we have
\begin{align}
&\int_{t_1+\eps}^{T} \Delta_x f(X_t) z_t a(t) dt \nonumber\\
&= \int_{t_1+\eps}^{T} \Delta_x f(X_{t_1+\eps}) z_{t_1+\eps} a(t) dt  +  \int_{t_1+\eps}^{T}  (m_t-m_{t_1+\eps}) a(t) dt  \nonumber\\
&= \Delta_x f(X_{t_1+\eps}) z_{t_1+\eps} \int_{t_1+\eps}^{T}  a(t) dt + \int_{t_1+\eps}^{T}  (m_t-m_{t_1+\eps}) a(t) dt \label{eq:t1T}
\end{align}

For $t \in [0,t_1)$, integrating again Equation (\ref{eq:sde_proof_delta_discret}) now over the interval $[0,t]$, we get
$$ \Delta_x f(X_t) z_t = \Delta_x f(X_0) + m_{t}.$$
Multiplying by $a \in \Gamma$ and integrating with respect to $t$ from 0 to $t_1-\eps$ give us
\begin{align}
&\int_{0}^{t_1-\eps} \Delta_x f(X_t) z_t a(t) dt \nonumber\\
&= \int_{0}^{t_1-\eps} \Delta_x f(X_0) a(t) dt +  \int_{0}^{t_1-\eps}  m_t a(t) dt  \nonumber\\
&=  \Delta_x f(X_0) \int_{0}^{t_1-\eps} a(t) dt +\int_{0}^{t_1-\eps}  m_t a(t) dt \label{eq:0t1}
\end{align}

Summing the two Equations (\ref{eq:t1T}) and (\ref{eq:0t1}), taking the expectation and using the fact $m$ is a martingale, we find
\begin{align*}
&\bE\left[\left(\int_0^{t_1-\eps} + \int_{t_1+\eps}^T\right) \Delta_x f(X_t) z_t a(t) dt \right] \\
&= \Delta_x f(X_0) \int_{0}^{t_1-\eps} a(t) dt + \Delta_x f(X_{t_1+\eps}) z_{t_1+\eps} \int_{t_1+\eps}^{T}  a(t) dt \\
&+ \bE\left[\int_{0}^{t_1-\eps}  m_t a(t) dt\right] + \bE\left[\int_{t_1+\eps}^{T}  (m_t-m_{t_1+\eps}) a(t) dt \right] \\
&=  \Delta_x f(X_0) \int_{0}^{t_1-\eps} a(t) dt + \Delta_x f(X_{t_1+\eps}) z_{t_1+\eps} \int_{t_1+\eps}^{T}  a(t) dt.
\end{align*}
Therefore, the result follows letting $\eps \to 0^+$ and applying the integration by parts formula and using that $a \in \Gamma$, which means $\int_{0}^{t_1} a(t) dt =1$ and $\int_{t_1}^{T} a(t) dt =0$.
\end{proof}

\begin{remark}
Comparing with Equation (\ref{eq:pi}), we conclude that Theorem \ref{thm:delta_discret} gives the same weight as in \cite{malliavin_greeks1}.
\end{remark}

\begin{remark}
Consider a contract $g(Y_T) = \phi(y_{t_1}, \ldots, y_{t_n})$, where $0 < t_1 < \cdots < t_n \leq T$ are fixed times and $\phi: \bR^n \longrightarrow \bR$. In the case of local volatility models, the assumption that $f$ is a discretely monitored functional in the previous Theorem is automatically satisfied as one can simply deduce from
$$f(Y_t) = \bE[\phi(x_{t_1}, \ldots, x_{t_n}) \ | \ Y_t],$$
and from Definition \ref{def:classif}.
\end{remark}

We would like to conclude this section observing that we were able to derive, using the techniques of functional It\^o calculus, the same results of \cite{malliavin_greeks1}, in which Malliavin calculus was used. Furthermore, the method implemented here enlightens the assumption that the derivative price needs to be a discretely monitored functional to employ Theorem \ref{thm:delta_discret}. Indeed, the main feature of such functionals is that they are locally weakly path-dependent in the interval $(t_i, t_{i+1})$ allowing us to apply the integration by parts formula in each of these interval.

One should also notice that, making the proper adaptations, a similar result to Theorem \ref{thm:gamma} holds true for discretely monitored functionals, since their Deltas are also discretely monitored functionals. Moreover, we should note that we could derive the equivalent of formula (\ref{eq:vega}) for discretely monitored functionals as well.

\subsubsection{Delayed Path-Dependent Functionals}

The argument presented in the proof of Theorem \ref{thm:delta_discret} can be generalized to the delayed path-dependent functionals. The next proposition states precisely the result. Define
$$\Gamma_s = \left\{ a \in L^2([0,T]) \ ; \ \int_0^{s} a(t) dt = 1 \mbox{ and } a(t) = 0, \mbox{ for } t \geq s\right\}.$$
\begin{proposition}\label{prop:delayed}
Fix a $t_1$-delayed path-dependent functional $f$ satisfying Assumptions \ref{assump:delta_reg} and consider $a \in \Gamma_{t_1}$. Thus,
\begin{align}
\Delta_x f(X_0) = \bE\left[g(X_T) \int_0^{t_1} \frac{a(t) z_t}{\sigma(X_t)} dw_t \right].\label{eq:delta_discret}
\end{align}
\end{proposition}

\begin{proof}
As before, by Equation (\ref{eq:mart_non_zero_lie}),
$$m_t = \Delta_x f(X_t) z_t - \Delta_x f(X_0) + \int_0^t \fL f(X_s) z_s ds.$$
Multiplying by any $a \in \Gamma_{t_1}$, integrating with respect to $t$ and taking expectation, we find
\begin{align*}
\Delta_x f(X_0) &= \bE\left[\int_0^T a(t) \Delta_x f(X_t) z_t dt \right] + \bE\left[\int_0^T a(t) \int_0^t \fL f(X_s) z_s ds dt \right] \\
&= \bE\left[\int_0^{t_1} a(t) \Delta_x f(X_t) z_t dt \right].
\end{align*}
Therefore, a simple application of the integration by parts formula yields the result.
\end{proof}

\begin{remark}
In the case of delayed path-dependent derivative, we have found the weight
$$\pi = \int_0^{t_1} \frac{a(t) z_t}{\sigma(X_t)} dw_t.$$
One should compare this formula with (\ref{eq:pi}).
\end{remark}

\begin{remark}
In the case when the Lie bracket is zero in $[u,s] \subset [0,T]$, we can adapt the proof above to find a similar expression of (\ref{eq:delta_discret}) for the Delta at time $u$, $\Delta_x f(X_u)$.
\end{remark}

\begin{remark}
Clearly, a discretely monitored functional is also delayed path-dependent, but it could be computationally advantageous to consider $a \in \Gamma$ instead of $a \in \Gamma_{t_1}$.
\end{remark}

\begin{example}\label{ex:forward_start}

Consider the following contract
$$g(X_T) = \left( x_T -  \frac{1}{T-t_1}\int_{t_1}^T x_u du\right)^+,$$
where $0 < t_1 < T$. This derivative is called \textit{forward-start floating-strike Asian call option}, see \cite{gobet_revisiting_greeks} for more details. We assume $x$ follows the Black--Scholes model with $r=0$, $dx_t = \sigma x_t dw_t$, where $\sigma > 0$. Hence, one can easily deduce that, for $t < t_1$, $f(Y_t) = \bE[g(X_T) \ | \ Y_t]$ depends only of $y_t$. Therefore, $f$ is a $t_1$-delayed path-dependent functional.

Applying Proposition \ref{prop:delayed}, we find
\begin{align}
\Delta_x f(X_0) = \bE\left[g(X_T)\int_0^{t_1} \frac{a(t) z_t}{\sigma x_t} dw_t \right].\label{eq:delta_asian}
\end{align}
Consider then the weight
$$\pi = \int_0^{t_1} \frac{a(t) z_t}{\sigma x_t} dw_t,$$
and further notice that in this model the tangent process satisfies $z_t = x_t/x_0$. Hence,
$$\pi = \frac{1}{\sigma x_0} \int_0^{t_1} a(t) dw_t \sim N\left(0 , \frac{1}{\sigma^2 x_0^2} \int_0^{t_1} a^2(t) dt \right).$$
One can show that the choice $a \equiv 1/t_1$ attains minimum variance for $\pi$ over $\Gamma_{t_1}$. Then,
$$\pi = \frac{w_{t_1}}{t_1 \sigma x_0}.$$
Considering the parameters given in Table \ref{tab:Parameters}, we find the results presented in Table \ref{tab:PriceDelta} and in Figure \ref{fig:mc_fwd_asian}. We compare the estimation using the weight we have derived with the finite difference approach (i.e. $\Delta_x f(X_0) \approx (f(X_0^h) - f(X_0^{-h}))/(2h)$, for a small $h$). We see that we achieve smaller standard error (and fastest convergence) using Equation (\ref{eq:delta_asian}).
\begin{table}[h!]
\begin{center}
 \begin{tabular}{c c}
 \hline Parameter & Value  \\
 \hline
 $X_0$ & 100 \\
 $\sigma$ & 0.2 \\
 $t_1$ & 0.2 \\
 $T$ & 1 \\
 \hline
 \end{tabular}
\caption{Parameters of the example on forward-start floating-strike Asian call options.} \label{tab:Parameters}
 \begin{tabular}{c c c}
 \hline
 \empty & Mean & Standard Error  \\
 \hline
 $f(X_0)$ & 3.5329 &  0.0200 \\
 $\Delta_x f(X_0)$ & 0.03607 & 0.00258 \\
 $\Delta_x f(X_0)$ (FD) & 0.04055 & 0.01409 \\
 \hline
 \end{tabular}
 \caption{Monte Carlo Estimation of the Price and Delta of a forward-start floating-strike Asian call option.} \label{tab:PriceDelta}
 \end{center}
\end{table}
\begin{figure}[h!]
\begin{center}
\includegraphics[scale = 0.5]{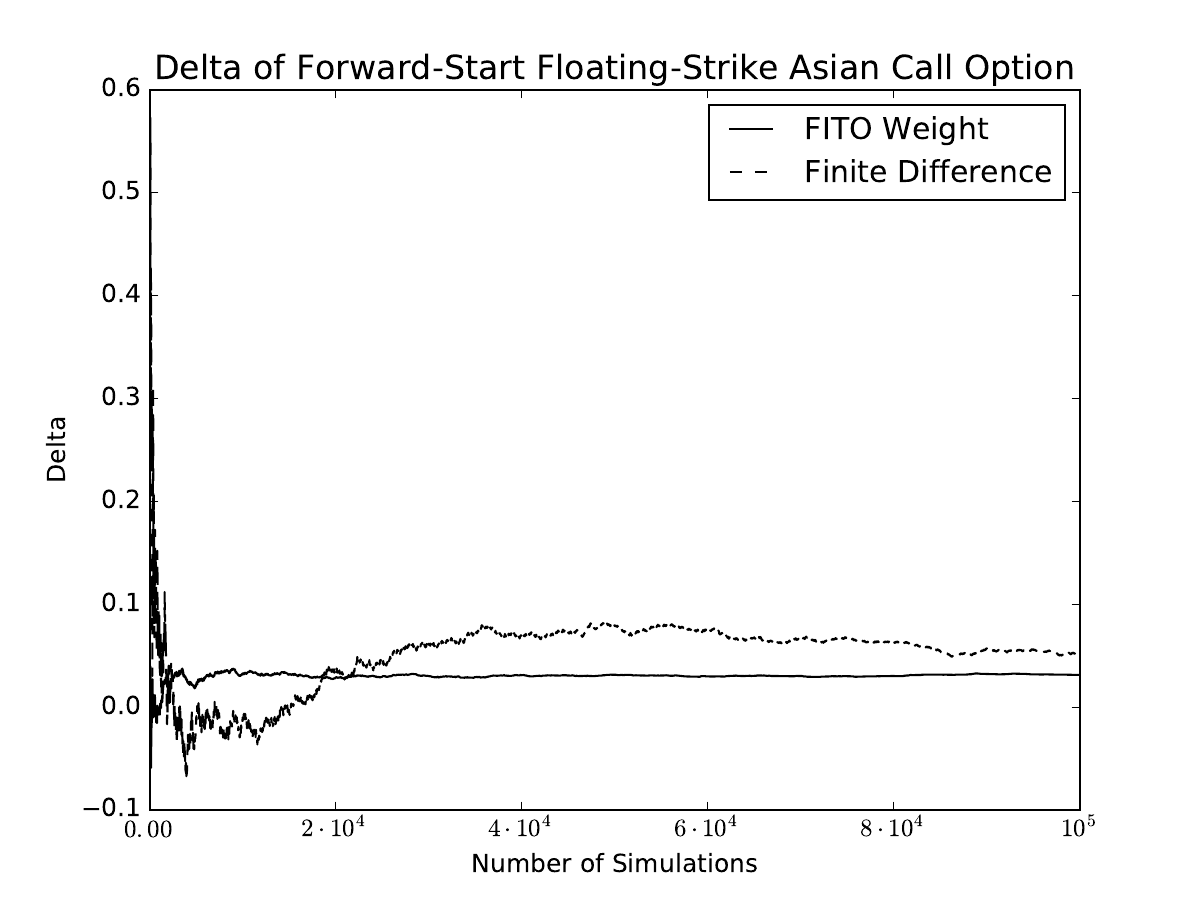}
\caption{Convergence Plot of the Monte Carlo Method to Compute $\Delta_x f(X_0)$ - Comparison with Finite Difference Method.}
\label{fig:mc_fwd_asian}
\end{center}
\end{figure}
\end{example}
\clearpage
\section{Conclusion and Future Research}

We have introduced an instantaneous measure of path-dependence using the functional It\^o calculus framework introduced in Dupire's influential work \cite{fito_dupire}. This measure is defined as the Lie bracket of the time and space functional derivatives. We then proposed a classification of functionals by their degree of path-dependence. Furthermore, for functionals with less severe path-dependence structures, called here locally weakly path-dependent, we studied the weighted-expectation formulas for the Delta, Gamma and Vega. In the case of a strong path-dependent functional, we were able to understand the impact of the Lie bracket on its Delta. Numerical examples of the theory were also presented.

Further research will be conducted to analyze the case of strong path-dependence. In particular, an explicit description of the adjoint of the time functional derivative $\Delta_t$. Moreover, another interesting direction for additional research is the case of non-linear prices of path-dependent derivatives and the computation of their Greeks; see for instance \cite{ma_zhang_01_aap, fito_zhang_1}

\subsubsection*{Acknowledgements}
\small
Firstly and more importantly, we would like to thank B. Dupire for proposing such interesting problem and all the enlightened discussions. Without him, this work would not be possible. We thank J.-P. Fouque and S. Corlay for all the insightful comments. We also thank the anonymous referees for their careful reading of the paper. Their comments improved the paper tremendously.

Part of this research was conducted while Y. F. Saporito was supported by Fulbright grant 15101796 and the CAPES Foundation, Ministry of Education of Brazil, Bras\'ilia, DF 70.040-020, Brazil.

The research was carried out in part during the summer internship of 2011 at Bloomberg supervised by B. Dupire.
\normalsize

\begin{appendices}

\section{Topological Support of SDEs with Functional Coefficients\protect\footnote{Compared to the published version: added the necessary assumption \textit{bounded from above} on the volatility and corrected the argument at the end of this appendix.}}\label{sec:topological_support}

As we have stated in Remark \ref{rmk:stroock_varadhan}, the Stroock-Varadhan Support Theorem guarantees that, under mild assumptions of the coefficients, the topological support of a diffusion will be the space of continuous spaces. However, to the best of our knowledge, this theorem is not available for the case of path-dependent coefficients. In what follows, we will discuss a simple assumption on the functional $\sigma$ that guarantees that the support of the process $x$ contains all the continuous paths.

Firstly, notice that we may remove the drift of SDE (\ref{eq:PD_vol}) by considering the discounted price, $e^{-rt}x_t$. In the case of the general SDE (\ref{eq:sde}), we may remove the drift $\mu$ by applying the Girsanov Theorem (assuming, for example, that $(\mu(X_t)/\sigma(X_t))_{t \in [0,T]}$ satisfies the Novikov's condition). Let us consider then the driftless SDE
$$dx_t = \sigma(X_t)dw_t.$$
The simple assumption on $\sigma$ that guarantees that the topological support of $x$ contains all the continuous functions in $[0,T]$ starting at $x_0$ is that $\sigma$ is \textit{bounded from above and from below}, i.e. $0 < a \leq \sigma(Y_t) \leq A < +\infty$, for all $Y_t \in \Lambda$. We wish to show that, given any $Z_T$ continuous and starting at $x_0$, for any $\delta > 0$,
$$\bP(\|X_T - Z_T\|_{\infty} < \delta) > 0.$$
By density arguments, it is sufficient to prove the inequality above for $Z_T$ of the form
$$z_t = x_0 + \int_0^t \phi_s ds,$$
for continuous $\Phi_T = (\phi_s)_{s \in [0,T]}$. Notice then
\begin{align*}
\|X_T - Z_T\|_{\infty} &= \sup_{t \in [0,T]} \left|\int_0^t \sigma(X_s)dw_s  - \int_0^t \phi_s ds\right| \\
&= \sup_{t \in [0,T]} \left|\int_0^t \sigma(X_s)d\left(w_s - \frac{\phi_s}{\sigma(X_s)} ds\right) \right|.
\end{align*}
Since $(\phi_t/\sigma(X_t))_{t \in [0,T]}$ is bounded ($\Phi_T$ is continuous and $\sigma$ is bounded from below), we may consider the equivalent probability measure $\bP^\phi$:
$$\left.\frac{d\bP^\phi}{d\bP}\right|_{\cF_t} = \exp\left\{\int_0^t \frac{\phi_s}{\sigma(X_s)} dw_s - \frac{1}{2} \int_0^t \frac{\phi_s^2}{\sigma^2(X_s)} ds \right\}.$$
Hence
\begin{align*}
&\bP\left(\sup_{t \in [0,T]} \left|\int_0^t \sigma(X_s)d\left(w_s - \frac{\phi_s}{\sigma(X_s)} ds\right) \right| < \delta \right) > 0 \Leftrightarrow \\
&\bP^\phi\left(\sup_{t \in [0,T]} \left|\int_0^t \sigma(X_s)d\left(w_s - \frac{\phi_s}{\sigma(X_s)} ds\right) \right| < \delta \right) > 0.
\end{align*}
Moreover, under $\bP^\phi$, the process $\left(w_t - \int_0^t \frac{\phi_s}{\sigma(X_s)} ds\right)_{t \in [0,T]}$ is a Brownian motion. Therefore, we have shown that
$$\bP(\|X_T - Z_T\|_{\infty} < \delta) > 0 \Leftrightarrow \bP\left(\sup_{t \in [0,T]} \left|\int_0^t \sigma(X_s)dw_s \right| < \delta \right) > 0$$
In order to show the right-hand side of the equivalence above, we would like to apply the Dambis-Dubins-Scharwz Theorem and to do this we will extend the process $x$ to the time domain $[0,+\infty)$ as follows:
$$x_t = x_T + a(w_t - w_T), \mbox{ for } t > T.$$
Notice that $(x_t)_{t \in [0, +\infty)}$ is a continuous local martingale with $\langle x \rangle_t$ almost surely strictly increasing and continuous in $t$ with $\langle x \rangle_\infty = +\infty$. Therefore, by the Dambis-Dubins-Scharwz Theorem, there exists a Brownian motion $(b_t)_{t \in [0,+\infty)}$ such that $x_t - x_0 = \int_0^t \sigma(X_s)dw_s = b_{\langle x \rangle_t}$. Thus, since $\langle x \rangle_T \leq A^2  T$,
$$\sup_{t \in [0,T]} \left|\int_0^t \sigma(X_s)dw_s \right| = \sup_{t \in [0,\langle x \rangle_T]} |b_t| \leq \sup_{t \in [0, A^2  T]} |b_t|,$$
which implies
$$\bP\left(\sup_{t \in [0,T]} \left|\int_0^t \sigma(X_s)dw_s \right| < \delta \right) \geq \bP\left(\sup_{t \in [0, A^2  T]} |b_t | < \delta\right) > 0,$$
as desired.

\section{Stochastic Integrals and Quadratic Variations}\label{sec:stoch_int_quad_var}

An important functional we would like to consider in the context of the functional It\^o calculus is the \textit{quadratic variation}. The first difficulty in this task is that this functional cannot be continuous with respect to the $d_{\Lambda}$ metric. In fact, for any $\eps > 0$, consider a process $(x_t)_{t \geq 0}$ starting at zero that is a Brownian motion in the strip $[-\eps,\eps]$ and reflects once it touches either barrier $-\eps$ or $\eps$. This process clearly satisfies $\| X_t \|_{\infty} \leq \eps$ and $\langle x \rangle_t = t$, for any $t \geq 0$, showing that the path $X_t$ is uniformly close to 0 with arbitrary quadratic variation.

Moreover, if we intuitively define $f(Y_t)$ as the quadratic variation of the path $Y_t$, we would face complications regarding the existence of this functional in $\Lambda$ and the choice of the sequence of partitions used to compute such quadratic variation. For instance, there exists a sequence of partitions that generates infinite quadratic variation for the Brownian motion.

Nonetheless, there are several ways to consider the quadratic variation functional. Here we will consider the framework of the Bichteler-Karandikar pathwise integral, see \cite{bichteler_stochastic_int} or \cite{karandikar_path_integral} for instance, where it is possible to consider a weaker continuity assumption on the functionals and extend the Functional It\^o Formula to this case. This was done in \cite{fito_extension_ito_formula} and we forward the reader there for the formal definitions and results below.

Consider the space of smooth functionals defined in the aforesaid reference, $\cC^{1,2}$.  This space extends $\bC^{1,2}$ by weakening the $\Lambda$-continuity assumption. For now, it is only necessary to know that $\bC^{1,2} \subset \cC^{1,2}$ and that the Functional It\^o Formula, Theorem \ref{thm:fif}, holds for functionals in $\cC^{1,2}$.

We now describe the Bichteler-Karandikar approach to define the pathwise stochastic integral. They proved there exists an operator $I: \Lambda_T \times \Lambda_T \longrightarrow \Lambda_T$ such that for any filtered probability space $(\Omega', \cF', \cF'_t, \bP')$, any semimartingale $x$ and any adapted, c\`adl\`ag process $z$, both in this probability space, satisfies
$$I(Z_T(\omega), X_T(\omega))(t) = \left(\int_0^t z_{s-} dx_s \right)(\omega) \quad \bP\mbox{-a.s.}$$


Now, fix a functional $h$ satisfying certain regularity requirements stated in \cite{fito_extension_ito_formula}. Then, there exists a functional
$$I_h : \Lambda \longrightarrow \bR$$
such that

\begin{enumerate}

\item $I_h \in \cC^{1,2}$;

\item $\ds I_h(X_t) = \int_0^t h(X_{s-}) dx_s$, for any continuous semimartingale $x$;

\item moreover, $\Delta_t I_h = 0$, $\Delta_x I_h(Y_t) = h(Y_{t-})$ and $\Delta_{xx} I_h = 0$.

\end{enumerate}

\noindent Here, the path $Y_{t-}$ is given by
\begin{align}
Y_{t-}(u) = \left\{
\begin{array}{ll}
  y_u,           &\mbox{ if } \quad u < t, \\
  y_{t-} = \lim_{u \to t^-} y_u, &\mbox{ if } \quad u = t.
\end{array}
\right. \label{eq:path_x_t-}
\end{align}

Furthermore, based on the well-known identity for semimartingales,
$$\langle x \rangle_t = x_t^2 - 2\int_0^t x_{s-} dx_s,$$
and since the pathwise definition of the stochastic integral is set, the \textit{pathwise quadratic variation} is defined by the identity
\begin{align}
QV(Y_t) = y_t^2 - 2 I_l(Y_t),
\end{align}
where the functional $l : \Lambda \longrightarrow \bR$ is given by $l(Y_t) = y_t$. From this, one can easily show

\begin{enumerate}

\item $QV \in \cC^{1,2}$;

\item $\ds QV(X_t) = \langle x \rangle_t$, for any continuous semimartingale $x$;

\item moreover, $\Delta_t QV = 0$, $\Delta_x QV(Y_t) = 2(y_t - y_{t-})$ and $\Delta_{xx} QV = 2$.

\end{enumerate}

We can then compute the Lie bracket of the stochastic integral and the quadratic variation functionals:
\begin{align*}
\fL I_h(Y_t) &= \left\{
\begin{array}{ll}
-\Delta_t h(Y_t) , &\mbox{ if } \Delta y_t = 0, \\ \\
\nexists, &\mbox{ if } \Delta y_t \neq 0,
\end{array}
\right.\\
\fL QV(Y_t) &= \left\{
\begin{array}{ll}
0 , &\mbox{ if } \Delta y_t = 0, \\ \\
\nexists, &\mbox{ if } \Delta y_t \neq 0,
\end{array}
\right.
\end{align*}
where $\Delta y_t = y_t - y_{t-}$ is the jump of $Y$ at time $t$.

Another functional we will be interested in is the pathwise version of the Dola\'eans-Dade exponential:
\begin{align}
E(Y_t) = \exp\left\{y_t -\frac{1}{2} QV(Y_t)\right\} \prod_{0 < s \leq t} (1 + \Delta y_s) \exp\left\{-\Delta y_s + \frac{1}{2}(\Delta y_s)^2\right\},
\end{align}
see \cite{protter05}. If $x$ is a continuous semimartingale, one can easily see that $E(X_t) = \exp\left\{x_t -\frac{1}{2} \langle x \rangle_t\right\}$. To compute the functional derivatives of $E$, notice that
\begin{align*}
E(Y_t) &= (1 + \Delta y_t) \exp\left\{-\Delta y_t + \frac{1}{2}(\Delta y_t)^2\right\} \\
&\exp\left\{y_t -\frac{1}{2} QV(Y_t)\right\} \prod_{0 < s < t} (1 + \Delta y_s) \exp\left\{-\Delta y_s + \frac{1}{2}(\Delta y_s)^2\right\}.
\end{align*}
Therefore, it is easy to conclude that
\begin{enumerate}

\item $\Delta_t E(Y_t)  = 0$;

\item $\ds \Delta_x E(Y_t) = \frac{1}{1 + \Delta y_t} E(Y_t)$ and $\ds \Delta_{xx} E(Y_t) = 0$.

\end{enumerate}

As we will see, we would like to compute the functional derivative of $f(Y_t) = E(I_h(Y)_t)$, where $I_h(Y)_t$ is the path $(I_h(Y_s))_{s \in [0,t]}$. Therefore, a chain rule argument allows us to write
\begin{align*}
&\Delta_t f(Y_t)  = 0, \ \Delta_x f(Y_t) = \frac{1}{1 + \Delta y_t} E(Y_t) h(Y_{t-}) \mbox{ and } \Delta_{xx} f(Y_t) = 0.
\end{align*}

\section{Localization Argument for Theorem \ref{thm:delta_0}}\label{sec:proof_delta_0}

By the assumptions on $f$ and $\sigma$, the process $m$, given in Equation (\ref{eq:martingale}), is a local martingale. Denote the integrand that defines $m$ by $(v_t)_{t \in [0,T]}$ and consider the sequence of stopping times
$$\tau_n = \inf\left\{t \in [0,T] \ ; \ \int_0^t (1 + (\Delta_x f)^2(X_s) + v_s^2) d\langle x \rangle_t \geq n\right\},$$
Then, $\tau_n \to T$ $\bP$-a.s, as $n \to +\infty$, and $(x_{t \wedge \tau_n})_{t \in [0,T]}$, $(m_{t \wedge \tau_n})_{t \in [0,T]}$, $(f(X_{t \wedge \tau_n}))_{t \in [0,T]}$ and
$$\int_0^{t \wedge \tau_n} \Delta_x f(X_s)dx_s$$
are proper martingales, for every $n \in \bN$. Moreover, $f \in \cD_n$, where $\cD_n = \cD_{x_{\cdot \wedge \tau_n}}$.

By the arguments shown in the proof of Theorem \ref{thm:delta_0}, where we have assumed that $f \in \cD_x$ and that $x$ and $m$ were proper martingales, we conclude that, for each $n \in \bN$,
$$\Delta_x f(X_{t \wedge \tau_n})z_{t \wedge \tau_n} = \Delta_x f(X_0) + m_{t \wedge \tau_n}.$$
Then $(\Delta_x f(X_t) z_t)_{t \in [0,T]}$ is clearly a local martingale and $(\tau_n)_{n \in \bN}$ is a localizing sequence for it. Now, integrating with respect to $t$, we get
$$\int_0^T \Delta_x f(X_{t \wedge \tau_n}) z_{t \wedge \tau_n} dt = \Delta_x f(X_0) T + \int_0^T m_{t \wedge \tau_n}  dt.$$


Hence, following the same steps performed in the proof of Theorem \ref{thm:delta_0}, we find
$$\Delta_x f(X_0) = \bE\left[f(X_{T \wedge \tau_n}) \frac{1}{T}  \int_0^{T \wedge \tau_n} \frac{z_t}{\sigma(X_t)} dw_t\right].$$
Notice now that $f(X_{T \wedge \tau_n}) = \bE[g(X_T) \ | \ \cF^x_{T \wedge \tau_n}]$, which implies
\begin{align*}
\Delta_x f(X_0) = \bE\left[g(X_T) \frac{1}{T}  \int_0^{T \wedge \tau_n} \frac{z_t}{\sigma(X_t)} dw_t\right].
\end{align*}
Deploying the Cauchy-Schwarz inequality, It\^o's Isometry and using the fact that $g(X_T) \in L^2$, we find
\begin{align*}
&\left|\Delta_x f(X_0) - \bE\left[g(X_T) \frac{1}{T}  \int_0^T \frac{z_t}{\sigma(X_t)} dw_t\right] \right| \leq \bE\left[\left|g(X_T) \frac{1}{T}  \int_{T \wedge \tau_n}^T \frac{z_t}{\sigma(X_t)} dw_t \right|\right]\\
&\leq \|g(X_T)\|_{L^2} \frac{1}{T} \left(\bE\left[ \int_{T \wedge \tau_n}^T \frac{z_t^2}{\sigma^2(X_t)} dt \right]\right)^{1/2} \stackrel{n \to +\infty}{\longrightarrow} 0,
\end{align*}
yielding the result.

\end{appendices}

\bibliographystyle{plain}

\begin{thebibliography}{99}

\bibitem{bichteler_stochastic_int}
K.~Bichteler.
\newblock Stochastic {I}ntegration and ${L}^p$ {T}heory of {S}emimartingales.
\newblock {\em Ann. Probab.}, 9(1):49--89, 1981.

\bibitem{rama_cont_fito_change_variable}
R.~Cont and D.-A. Fourni\'e.
\newblock {C}hange of {V}ariable {F}ormulas for {N}on-{A}nticipative
  {F}unctional on {P}ath {S}pace.
\newblock {\em J. Funct. Anal.}, 259(4):1043--1072, 2010.

\bibitem{rama_cont_fito_formula}
R.~Cont and D.-A. Fourni\'e.
\newblock A {F}unctional {E}xtension of the {I}t\^o {F}ormula.
\newblock {\em C. R. Math. Acad. Sci. Paris}, 348(1):57--61, 2010.

\bibitem{rama_cont_fito_mart}
R.~Cont and D.-A. Fourni\'e.
\newblock {F}unctional {I}t\^o {C}alculus and {S}tochastic {I}ntegral
  {R}epresentation of {M}artingales.
\newblock {\em Ann. Probab.}, 41(1):109--133, 2013.

\bibitem{dupire94}
B.~Dupire.
\newblock Pricing with a {S}mile.
\newblock {\em Risk Magazine}, 7:18--20, 1994.

\bibitem{fito_dupire}
B.~Dupire.
\newblock Functional {I}t\^o {C}alculus.
\newblock 2009.
\newblock Available at SSRN: \url{http://ssrn.com/abstract=1435551}.

\bibitem{fito_zhang_1}
I.~Ekren, C.~Keller, N.~Touzi, and J.~Zhang.
\newblock On {V}iscosity {S}olutions of {P}ath {D}ependent {PDE}s.
\newblock {\em Ann. Probab.}, 42(1):204--236, 2014.

\bibitem{fito_touzi_ppde1}
I.~Ekren, N.~Touzi, and J.~Zhang.
\newblock {V}iscosity {S}olutions of {F}ully {N}onlinear {P}arabolic {P}ath
  {D}ependent {PDE}s: {P}art {I}.
\newblock {\em Ann. Probab}, 44(2):1212--1253, 2016.

\bibitem{fito_touzi_ppde2}
I.~Ekren, N.~Touzi, and J.~Zhang.
\newblock {V}iscosity {S}olutions of {F}ully {N}onlinear {P}arabolic {P}ath
  {D}ependent {PDE}s: {P}art {II}.
\newblock {\em Ann. Probab}, 44(4):2507--2553, 2016.

\bibitem{foschi2008path}
P.~Foschi and A.~Pascucci.
\newblock Path {D}ependent {V}olatility.
\newblock {\em Decis. Econ. Finance}, 31(1):13--32, 2008.

\bibitem{fournie_cont_thesis}
D.-A. Fourni\'e.
\newblock {\em Functional {I}t\^o {C}alculus and {A}pplications}.
\newblock PhD thesis, Columbia University, 2010.

\bibitem{malliavin_greeks2}
E.~Fourni\'e, J.-M. Lasry, J.~Lebuchoux, and P.-L. Lions.
\newblock {A}pplications of {M}alliavin {C}alculus to {M}onte {C}arlo {M}ethods
  in {F}inance, {II}.
\newblock {\em Finance Stoch.}, 5(2):201--236, 2001.

\bibitem{malliavin_greeks1}
E.~Fourni\'e, J.-M. Lasry, J.~Lebuchoux, P.-L. Lions, and N.~Touzi.
\newblock {A}pplications of {M}alliavin {C}alculus to {M}onte {C}arlo {M}ethods
  in {F}inance.
\newblock {\em Finance Stoch.}, 3(4):391--412, 1999.

\bibitem{gobet_revisiting_greeks}
E.~Gobet.
\newblock {R}evisiting the {G}reeks for {E}uropean and {A}merican {O}ptions.
\newblock {\em Proceedings of the ``International Symposium on Stochastic
  Processes and Mathematical Finance" at Ritsumeikan University, Kusatsu,
  Japan. Edited by J. Akahori, S. Ogawa, S. Watanabe.}, pages 53--71, 2004.

\bibitem{gobet_malliavin_barrier_lookback}
E.~Gobet and A.~Kohatsu-Higa.
\newblock {C}omputations of {G}reeks for {B}arrier and {L}ookback {O}ptions
  {U}sing {M}alliavin {C}alculus.
\newblock {\em Electron. Commun. Probab.}, 8:51--62, 2003.

\bibitem{guyon_path_vol}
J.~Guyon.
\newblock Path {D}ependent {V}olatility.
\newblock {\em Risk Magazine}, Sep. 2014.

\bibitem{complete_sv_rogers}
D.~G. Hobson and L.C.G. Rogers.
\newblock Complete {M}odels with {S}tochastic {V}olatility.
\newblock {\em Math. Finance}, 8(1):27--48, 1998.

\bibitem{moore_interchanging_limit}
Z.~Kadelburg and M.~Marjanovi\'c.
\newblock Interchanging {T}wo {L}imits.
\newblock {\em The Teaching Of Mathematics}, VIII:15--29, 2005.

\bibitem{karandikar_path_integral}
R.~L. Karandikar.
\newblock On {P}athwise {S}tochastic {I}ntegration.
\newblock {\em Stochastic Process. Appl.}, 57(1):11--18, 1995.

\bibitem{lee_quad_var_derivative}
R.~Lee.
\newblock Realized volatility options.
\newblock In R.~Cont, editor, {\em Encyclopedia of Quantitative Finance}.
  Wiley, 2010.

\bibitem{ma_zhang_01_aap}
J.~Ma and J.~Zhang.
\newblock {R}epresentation {T}heorems for {B}ackward {S}tochastic
  {D}ifferential {E}quations.
\newblock {\em Ann. Appl. Probab.}, 12(4):1390--1418, 2002.

\bibitem{malltelm05}
P.~Malliavin and A.~Thalmaier.
\newblock {\em Stochastic {C}alculus of {V}ariations in {M}athematical
  {F}inance}.
\newblock Springer, 2005.

\bibitem{nualart_malliavin_book}
D.~Nualart.
\newblock {\em {M}alliavin {C}alculus and {R}elated {T}opics}.
\newblock Springer, second edition, 2006.

\bibitem{fito_extension_ito_formula}
H.~Oberhauser.
\newblock {A}n extension of the {F}unctional {I}t\^o {F}ormula under a {F}amily
  of {N}on-dominated {M}easures.
\newblock {\em Stoch. Dyn.}, 16(4), 2016.

\bibitem{fito_bsde_peng}
S.~Peng and F.~Wang.
\newblock {BSDE}, {P}ath-dependent {PDE} and {N}onlinear {F}eynman-{K}ac
  {F}ormula.
\newblock {\em Science China Mathematics}, 59(1):19--36, 2016.
\newblock Available at arXiv: \url{http://arxiv.org/abs/1108.4317}.

\bibitem{pinsky95}
R.~G. Pinsky.
\newblock {\em Positive {H}armonic {F}unctions and {D}iffusion}.
\newblock Cambridge University Press, 1995.

\bibitem{protter05}
P.~E. Protter.
\newblock {\em Stochastic {I}ntegration and {D}ifferential {E}quations}.
\newblock Springer, second edition, 2005.

\bibitem{rogerswilliams}
L.C.G. Rogers and D.~Williams.
\newblock {\em Diffuions, {M}arkov {P}rocesses and {M}artingales}.
\newblock Cambridge Mathematical Library, second edition, 2000.

\bibitem{wilmott_quant_fin}
Paul Wilmott.
\newblock {\em Paul {W}ilmott on {Q}uantitative {F}inance}.
\newblock Wiley, second edition, 2006.

\end{thebibliography}

\end{document}